\documentclass[11pt, twocolumn, floatfix, superscriptaddress, aps, prb, raggedbottom]{revtex4-2}
\usepackage[a4paper, margin=2cm]{geometry}
\usepackage{amscd,amsmath,amssymb,amsthm,mathtools, braket}
\usepackage{enumitem,graphicx,hyperref,xcolor}
\usepackage{braket,cleveref,comment,csquotes,dsfont,ifthen,setspace,wrapfig}
\usepackage[normalem]{ulem}

\newtheorem{theorem}{Theorem}[section]
\newtheorem{lemma}[theorem]{Lemma}
\newtheorem{corollary}[theorem]{Corollary}
\newtheorem{definition}{Definition}[section]
\newtheorem{proposition}{Proposition}[section]

\newcommand{\comm}{\mathrm{Comm}}

\newcommand{\SWAP}{\mathrm{SWAP}}
\newcommand{\FSWAP}{\mathrm{FSWAP}}
\newcommand{\Id}{\mathds{1}}
\newcommand{\cstar}{$\mathrm{C}^*$-algebra }

\DeclareMathOperator{\tr}{\mathbf{Tr}}
\DeclareMathOperator{\Ad}{Ad}

\setlist[itemize]{leftmargin=3.5mm,itemsep=0.0mm}

\makeatletter
\newcommand{\whencolumns}[2]{\twocolumn@sw{#1}{#2}}
\makeatother
\begin{document}
\singlespacing
\begin{abstract}

Finding analytically tractable model systems for quantum ergodicity breaking in interacting systems remains a theoretical challenge. This paper presents an algebraic theory of robust ergodicity breaking in time-periodic unitary circuits under local causal constraints. Our model is based on constructing tripartite unitaries (we dub \enquote*{walls}) that permanently arrest local operator spreading. We show that the structure of the resulting causally independent subsystems can be understood rigorously through the invariance of embedded operator algebras (i.e.~super-operator symmetries). This formalism gives a natural identification of local conserved quantities and permits the generalisation to time-dependent dynamics. Using representation theory, the general form of unitaries exhibiting causal decoupling is derived from the automorphism group of the embedded algebra with links to quantum error-correcting codes. %Our construction gives conditions for no-signalling in tripartite unitaries which may be of independent interest. 
From the point of view of operator spreading, our theory is a minimal model for non-ergodic quantum-circuit dynamics, and we explore its effects on probes of many-body quantum chaos. We prove an entanglement area law due to causal constraints and discuss its stability against local measurements. In a random unitary ensemble with causally independent subsystems, we compare spectral correlations with those of the universal random matrix ensemble using the spectral form factor. Our results offer a rigorous understanding of locally constrained quantum dynamics from a quantum information perspective. %\lluis{[...understanding of local quantum dynamics?]} from a quantum information perspective.
\end{abstract}

%\title{Minimal model for confined operator dynamics in Floquet quantum circuits}
%\title{Minimal model for causal independence in time-periodic operator dynamics}
%\title{Causal independence in quantum operator dynamics}
\title{Causally constrained quantum operator dynamics}

\author{Marcell D. Kovács}
\affiliation{Department of Physics and Astronomy, University College London, United Kingdom}
\author{Christopher J. Turner}
\affiliation{Department of Physics and Astronomy, University College London, United Kingdom}
\author{Lluís Masanes}
\affiliation{Department of Computer Science, University College London, United Kingdom}
\affiliation{London Centre for Nanotechnology, University College London, United Kingdom}
\maketitle

\section{Introduction}

%%% 1. Introduce the general area of quantum chaos. The different perspectives and applications.

The study of quantum chaos with a quantum information lens has become a central theme in quantum many-body physics. In these systems, the dynamics of local operators explores all the available space, a phenomenon known as operator spreading and scrambling~\cite{LiebRobinson1972, Chen_2023, Roberts2017}. Chaotic dynamics also produces rapid growth of entanglement, correlations in the energy spectrum similar to those of random matrices~\cite{Wigner1967,Brezin1997} and eigenstate thermalisation~\cite{DAlessio2016review}.
Many of these phenomena are most easily understood in the context of random local quantum circuits, which often form an analytically tractable alternative to local Hamiltonian dynamics~\cite{Nahum2017, Nahum2018, Bertini2018, CurtvK2018}. 
To obtain a complete picture of quantum chaos, it is important to understand when some of the above phenomena are suppressed. In particular,
violations of ergodicity are of great interest due to the late-time recoverability of locally encoded quantum information. Systems where ergodicity is suppressed include the so-called integrable systems with complete sets of integrals of motion, and Anderson and many-body localisation due to disorder~\cite{Rigol2007, Anderson1958, Basko2006, Abanin2019}.
More recent examples include quantum scars~\cite{Serbyn2021review, Turner2018, Moudgalya2022review} and fragmentation~\cite{Moudgalya2022, Pai2019, Sala2020, Khemani2020, Moudgalya2022review}, which are systems where eigenstate thermalisation fails -- while other features of quantum chaos are retained.

%%% 2. Focus in on our research topic, maybe research questions and what our contribution is.

\begin{figure}
    \centering
\includegraphics[width=0.5\textwidth]{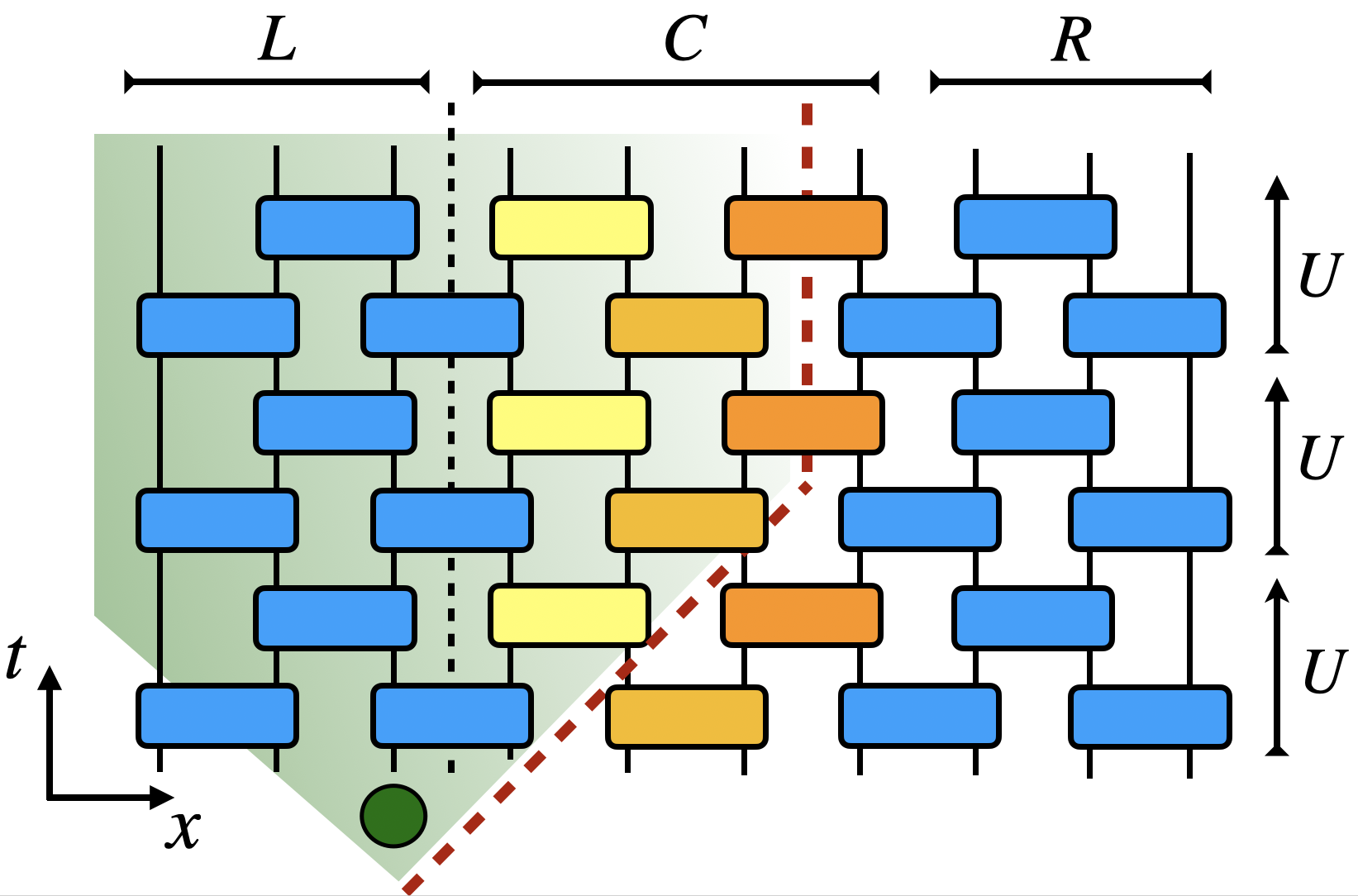}
    \caption{ \small We study the time-periodic evolution of local traceless operators in brickwork unitaries. A wall unitary tri-partitions the circuit and permanently obstructs the spreading of local operators. This leads to a bounded light cone (shown in red) and causal decoupling across spatial regions. This paper studies the general algebraic properties of unitaries exhibiting bounded light cones and their effect on many-body ergodicity-breaking.}
    \label{fig:bounded_light_cone}
\end{figure}

%%% 2a. Chaos first

In random quantum circuits composed of independently sampled gates, the rate at which the circuit approaches the globally random ensemble -- without locality structure -- can be quantified through the use of $k$-designs~\cite{Roberts2017, Brandao2016, Harrow2009, Haferkamp2022}.
Alternatively, in time-periodic random circuits, while the convergence to random matrices is inhibited, the dynamics can retain a high-degree of chaos \cite{Roberts2017,Pilatowsky-Cameo2024_CHSE_CUE, Mark_2024_prx} under mild spectral no-resonance conditions.
For global chaos to emerge from locally interacting structure, universal gatesets must be used, which raises the question whether the ergodic hypothesis can be violated with generic gatesets and to what extent operator scrambling can be locally retained.
%what are atypical quantum circuits built from universal gateset that violate the ergodic hypothesis and to what extent can operator scrambling be retained locally.
% In this paper, we set out to understand how operator spreading can be halted in time-periodic quantum circuits (that we dub \enquote*{bounded light cones}, as on \Cref{fig:bounded_light_cone}) as a model to understand globally non-ergodic while locally chaotic dynamics.
% This connects several perspectives across many-body physics and quantum information theory including localisation, fragmentation and causality. 
%\lluis{For global chaos to emerge in local circuits, universal gate sets must be be used, which raises the question of why some (atypical) quantum circuits built from universal gate sets violate the ergodic hypothesis and to what extent can operator scrambling be retained locally.}
In this paper, we set out to understand how operator spreading can be halted in time-periodic quantum circuits -- something we call \enquote*{bounded light cones} as depicted in \Cref{fig:bounded_light_cone}.
This is an interesting case of global non-ergodicity with locally chaotic dynamics, and connects several perspectives across many-body physics and quantum information theory including localisation, fragmentation and causality.
%}

%%% 2b. Localisation
The phenomenon of localisation, where initially encoded quantum information is retained at late times, was historically established in strongly disordered spin-chains \cite{Anderson1958, Basko2006, Abanin2019Colloquium:Entanglement}.
However, the difficulty of understanding non-integrable Hamiltonian systems lead to the search for circuit analogues to many-body localisation, for which finely tuned examples were found~\cite{Chandran2015, Farshi2022_1D, Farshi2022_2D, Bertini2022_strongly_localised_circuits, Kovacs2026}, and stability to perturbations examined~\cite{Kovacs2026}.
In light of these finely tuned examples, what is the most general class of quantum circuit which exhibits localisation?
Given that many-body localisation is supposed to be a robust phase of matter, to what extent can circuit localisation be robust to local perturbations~\cite{Kovacs2026}? 

%%% 2c. Fragmentation
In the presence of an unconventional symmetry, a system might split into exponentially many invariant sectors and be said to display fragmentation~\cite {Moudgalya2022review}.
This lets the information of which sector a system is prepared in to be preserved, and is notably robust against general perturbation~\cite{Stephen2024, Kovacs2026}.
Frequently, the underlying symmetry is a so-called \emph{higher-order} symmetry, such as dipole moment conservation -- which can be realised physically through a strong separation of energy scales, as in the thin-torus limit of the fractional quantum Hall effect~\cite {Zerba2025}, or through the Wannier-Stark effect~\cite{Schulz2019, vanNieuwenburg2019}.
The underlying symmetries can be uncovered with operator-theoretic techniques, which were first developed in time-independent Hamiltonian systems~\cite{Moudgalya2022}, then generalised to Clifford circuits~\cite{Kovacs2026} and open or dissipative systems~\cite{YahuiLi2023_OpenHSF, paszko2025operatorspacefragmentationintegrabilitypaulilindblad, Paszko2024_OpenSPT}.
We will show how local fragmentation can be analysed in periodic quantum circuits and find the unconventional superoperator symmetries which cause localisation.
We also give some indications on how this can be generalised to fully time-dependent or aperiodic systems.

%%% 2d. Causality
Operator algebras in quantum systems also formalises how influence between subsystems is constrained by causality in unitary evolution~\cite{Allen2017}.
Within the algebraic framework for quantum field theory \cite{Fewster2019, Haag1996, BuchholzFredenhagen2023, BrunettiEtAl2015}, causal independence is understood through the commutation structure of the operator algebra of observables.
On the quantum foundations side, characterising the compositional structure of unitaries that exhibit \textit{causal independence} between subsystems has been studied to construct faithful diagrammatic representations beyond the quantum circuit model using operator algebraic methods \cite{Lorenz2021, Ormrod2023}.
In these works, the connectivity of subsystems arise from the presumed causal constraints that can be saliently understood using the splitting of operator algebras.
%\chris{This sentence seems decoupled from the rest of the paragraph:}
%Casually decoupled subsystems also have implications for quantum field theory because they serve as models for space-like separated events outside of a relativistic light-cone \cite{DeFacio1975, BuchholzWichmann1986, Summers1990}.

The novelty of our approach is to use operator space techniques from quantum foundations to understand the effect of causal separation on probes of many-body dynamics. 
In particular, we prove that causal decoupling in time-periodic circuits is equivalent to the existence of super-operator symmetries. 
This enables us to characterise the unitaries that exhibit bounded light cones for a given symmetry and prove that the localisation leads to an area-law phase robust to sparse local perturbations. \textit{En route}, we generalise several technical results across the quantum information literature on no-signalling unitaries, which may be of independent interest.
In particular, our results show that no-signalling can happen in tri-partite unitaries with non-product structure, which extends the results of \cite{Beckman2001_causal_quatum_op}. 
Additionally, we extend the results of \cite{Lorenz2021, Ormrod2023} on the structure of unitary maps exhibiting causal independence to time-periodic evolutions.
While the bounded light cone problem was first considered by the authors in the context of Clifford circuits \cite{Kovacs2026}, this paper shows that non-universal gate sets are not necessary for operator localisation to occur and the analytical tractability can be maintained for a generic model of ergodicity breaking.

%%% 4. Outline
The paper is structured as follows. In \Cref{sec:bounded_light_cones}, we develop an algebraic theory of causal independence for tripartite unitaries, constructing invariant algebras and characterising the splitting of the many-body operator space. This enables the study of the structure of local conservation laws. In \Cref{sec:rep_theory}, we employ the representation theory of finite matrix algebras to characterize the brickwork unitary circuits exhibiting bounded light-cones, separating cases with Abelian and non-Abelian invariant algebras to show that the nature (and existence) of conserved quantities depends strongly on the commutativity of the algebra. Notably, the condition for causal decoupling to occur without any conserved quantities is that the invariant algebra is a single factor hosting a hidden quantum subsystem. Finally, \Cref{sec:dynamical_features} is dedicated to connecting the bounded-light cone phenomenology to probes of many-body ergodicity. We prove an entanglement area law in the circuit, consider its stability to local measurements and find the signature of ergodicity-breaking in spectral correlations through the spectral form factor. We conclude and comment on further generalisations of our circuit to infinite-dimensional systems as well as the hardness of verifying causally decoupled dynamics on quantum computers. 

%\whencolumns{}{\newpage}

\section{Bounded light cones in tripartite systems\label{sec:bounded_light_cones}}

\subsection{General setup}
\begin{figure}
    \centering
    \includegraphics[width=\linewidth]{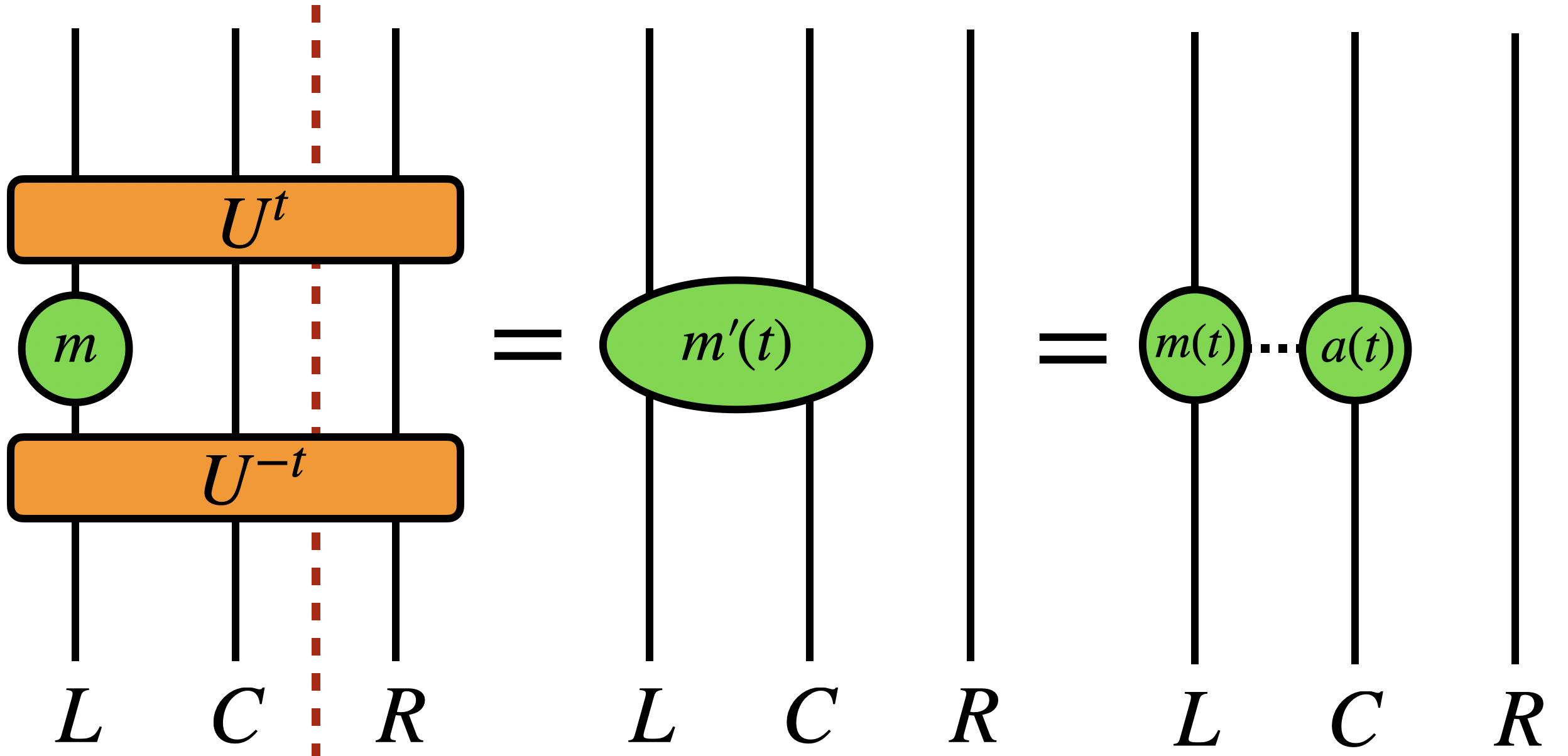}
    \caption{ \small Localisation of left-evolving operators in tensor network notation. Green operators correspond to traceless operators where dotted lines signify arbitrary superpositions. Elements of the full matrix algebra $m \in \mathcal{M}_L \otimes \Id_{CR}$ must localise under the wall evolution. The operator evolution is constrained to occur in a sub-algebra under \Cref{thm:left_right_closure}: $m'(t) \in \mathcal{M}_L \otimes \mathcal{A}_C \otimes \Id_R$ .}
    \label{fig:wall_condition}
\end{figure}
We consider the bounded light-cone problem for arrested operator spreading generically defined over a tripartite quantum system. Our goal in this section is to provide the algebraic characterisation of a local causal constraint which we will apply later in the brickwork circuit setting.

We consider a left system $L$, central system $C$ and right system $R$ of a tripartite Hilbert space $\mathcal{H}_{LCR}$. By assumption, these are finite-dimensional quantum systems with an associated algebra of complex matrices (\textit{full matrix algebra}) denoted by $\mathcal{M}_{LCR}$ that act on the Hilbert space. In terms of the subsystem algebras, we have the product form $\mathcal{M}_{LCR} =\mathcal{M}_L \otimes \mathcal{M}_C \otimes \mathcal{M}_R$. Sometimes we will use the shorthand $\mathcal{M}_L$ to mean $\mathcal{M}_L \otimes \Id_{CR}$ and similarly for the other algebras of local operators as appropriate. In \Cref{app:algebras}, we provide a more detailed elementary introductions for technical details.

As on Figure \ref{fig:bounded_light_cone}, we are concerned with the adjoint evolution of a local traceless operator $O_0$ non-trivially supported on a local subsystem at initial time: 
\begin{equation} O(t) = \mathrm{Ad}^t_U(O_0) = U^t O_0 U^{-t}.\end{equation}

\noindent We will refer to $U$ as the Floquet unitary. 
 The focus of this paper is \emph{strictly} localised operator evolution such that an initially local operator remains non-trivially supported on an enlarged, finite, subsystem at all times, despite the connectivity of the couplings between subsystems.
 
We formalise this with the concept of a \textit{wall} unitary on a tripartite quantum system. For convenience and generality, we first derive the associated localisation condition on unstructured unitaries. We define a \enquote*{wall} as a Floquet unitary which generates arrested local operator spreading in the corresponding adjoint map. The adjoint unitary evolution is an algebraic isomorphism since $\Ad_U$ is a linear operator preserving product structure, adjoint structure and dimension. As a result, we can recast the localisation condition in terms of the evolution of operator algebras too which will be useful to derive the structure of wall unitaries later.

%\lluis{Perhaps this is not the actual setup, but the "unstructured/general setup". Perhaps here, and not in the definition of walls, I would introduce the systems $L$ (left), $C$ (centre) and $R$ (right), alongside the Floquet unitary $U$. Since all this section and subsequent ones only considers this setup. Perhaps it is also useful to define the operator algebras $\mathcal M_L, \mathcal M_C, \mathcal M_R$. Also, it is not clear to me whether the systems $L,C,R$ are finite-dimensional or not.}

\begin{definition}[Walls]
  %\label{def:walls}
  Consider a unitary $U$ jointly acting on the systems $L$ (left), $C$ (centre) and $R$ (right). We say that $U$ is a \textit{left-wall}, if the induced adjoint evolution of any operator $m$ in $L$ remains non-trivially supported on $LC$ at all times. That is,  for each $t\geq 0$ there is $m'(t)$ in $LC$ such that 
  \begin{equation}
    \mathrm{Ad}_U^t \left( m_L \otimes \mathds{1}_{CR}  \right) = m'(t)_{LC} \otimes \mathds{1}_{R}\ .
  \end{equation}
  Equivalently, and in terms of operator algebras,
  \begin{equation}
    \mathrm{Ad}_U^t \left( \mathcal M_L \otimes \mathds{1}_{CR}  \right) \subseteq  \mathcal M_{LC} \otimes \mathds{1}_{R}\ ,
  \end{equation}
  for all $t\geq 0$. Analogously, $U$ is a \textit{right-wall} if   
    \begin{equation}
        \mathrm{Ad}_U^t \left(\mathds{1}_{LC}   \otimes \mathcal M_R \right) \subseteq \mathds{1}_{L} \otimes \mathcal M_{CR}\ ,
      \end{equation}
  for all $t\geq 0$.
  \label{def:walls}
\end{definition}

The same localisation condition was investigated in \cite{Kovacs2026} for Clifford circuits which we now generalise using the operator algebra framework.
The above definition is a generalisation of causal independence in the works \cite{Lorenz2021, Allen2018} where the causal decoupling if demanded for a single time-step evolution of the unitary. These works considered the compositional structure of unitary maps that satisfy causal constraints. By demanding that $L,R$ remain causally decoupled for arbitrary times, we are restricting the eigenspaces of the unitary to fulfil a localisation constraint. 

The causal independence of subsystems can be equivalently formulated in terms of commutation of their operator algebras. For this, we will extensively use the notion of a commutant algebra, $\comm(\mathcal{A_S})$, defined as the largest algebra of operators on $\mathcal{M}_S$ that commute with all elements of sub-algebra $\mathcal{A}_S \subseteq \mathcal{M}_S$ for a quantum system $S$. We refer to \Cref{app:commutants} for a detailed introduction.

\begin{lemma}[Commuting time-evolved algebras]
    The left wall condition for a unitary $U$ is equivalent to:
    \begin{equation}
        [\Ad_U^t(\mathcal{M}_L), \mathcal{M}_R] = 0 \text{ for all } t \geq 0.
        \label{eq:left_commutator}
    \end{equation}
    Similarly, the right-wall condition for a unitary $U$ can be written as:
    \begin{equation}
        [\mathcal{M}_L,\Ad_U^t(\mathcal{M}_R)] = 0 \text{ for all } t \geq 0,
        \label{eq:right_commutator}
    \end{equation}
    where the commutator relation between algebras, $[\mathcal{A}, \mathcal{B}] = 0$, denotes pairwise commutation between elements of $\mathcal{A}$ and $\mathcal{B}$.
    \label{lemma:walls_as_commutators}
\end{lemma}
\begin{proof}
    We make use of the double-commutant theorem of von-Neumann algebras, as detailed in \Cref{app:commutants}. The largest algebra of operators on $LCR$ commuting with $\mathcal{M}_L$ is $\comm(\mathcal{M}_L) = \mathcal{M}_{CR}$ and similarly $\comm(\mathcal{M}_R) = \mathcal{M}_{LC}$. Under \Cref{def:walls}, the left wall condition implies that the algebra $\Ad_U^t(\mathcal{M}_L)$ is a sub-algebra of $\mathcal{M}_{LC}$ that implies that $\Ad_U^t(\mathcal{M}_L) \subseteq \comm(\mathcal{M}_R)$. This implies the pairwise commutation in \Cref{eq:left_commutator}. An analogous proof establishes  \Cref{eq:right_commutator}.
\end{proof}

Using the previous lemma, we now establish the equivalence of left-wall and right-wall conditions. The results relies on the observation that the localised operator evolution for arbitrary times implies a restriction on the wall unitary map's eigenspaces.
\begin{figure}
    \centering
    \includegraphics[width=0.85\linewidth]{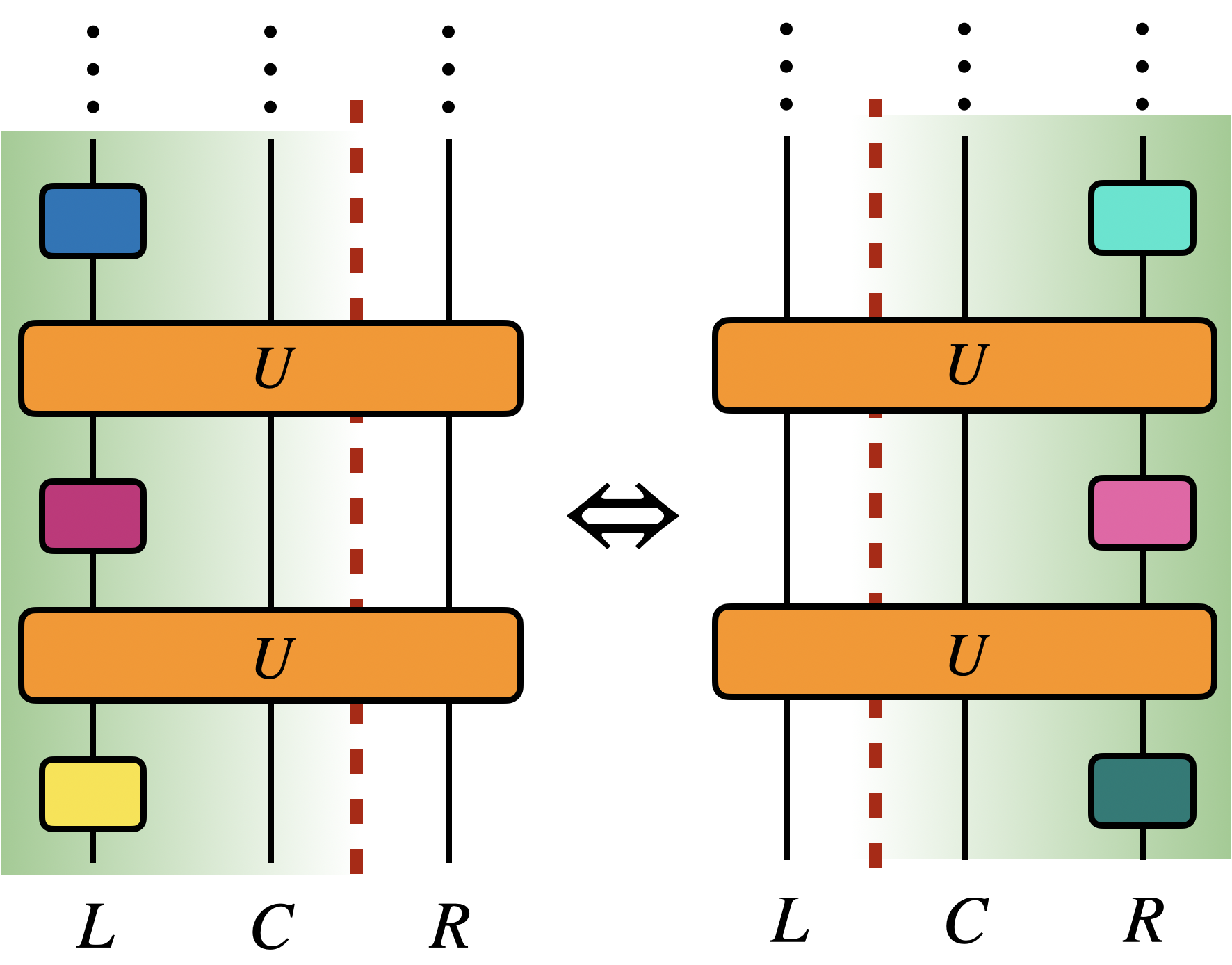}
    \caption{\small Stable embedding of a tripartite wall unitary into an arbitrary circuit environment under \Cref{thm:left_right_equiv}. The existence of a bounded light cone from the left-environment implies that of from the right inducing a splitting of operator space.}
    \label{fig:left_equiv_right}
\end{figure}
\begin{theorem}
    $U$ is a left-wall if and only if it is a right-wall.\label{thm:left_right_equiv}
\end{theorem}
\noindent We provide two equivalent proofs for completeness, one based on polynomials formed from the adjoint map and another based on the spectral decomposition.
\begin{proof}
    First, we note that the vector space of superoperators, such as $\Ad_U^{-1}$, is finite-dimensional.
    % CJT: If we can't assume minimal polynomials:
    % Hence, for any infinite sequence of such, there must be a linear dependence.
    % For the sequence of non-negative powers of $\Ad_U$ this provides an annihilating polynomial $f(x)$ such that $f(\Ad_U)=0$.
    % Choose this polynomial to be of minimal degree and monic such that it is the minimal polynomial of $\Ad_U$.
    Consequently, $\Ad_U^{-1}$ has a minimal polynomial $f(x)$ that annihilates it, i.e. $f(x) = 0$.
    $x$ does not divide $f(x)$, because otherwise there would exist a non-zero (by minimality of $f$) superoperator $A$ such that $\Ad_U^{-1} \circ A = 0$, however $\Ad_U^{-1}$ is invertible.
    Equivalently, this is because $\Ad_U^{-1}$ does not have a zero eigenvalue.
    % CJT: Or we could simply remark x only divides f(x) when 0 is an eigenvalue.
    We then use the extended Euclidean algorithm to find polynomials $s(x)$ and $t(x)$ such that $s(x) x + t(x) f(x) = \gcd(x, f(x)) = 1$.
    Evaluating this equation at $\Ad_U^{-1}$ reveals that $\Ad_U = s(\Ad_U^{-1})$, i.e. the evolution map is a polynomial in its inverse.

    If we take the right-wall conditions from \Cref{eq:right_commutator}, we may transfer the evolution to the left argument of the commutator in terms of non-positive powers,
    \begin{align}
      [\Ad_U^{-t}(\mathcal{M}_L),\mathcal{M}_R] = 0, && \text{for all $t \ge 0$.}
      \label{eq:negative_t_right_wall}
    \end{align}
    Using the linear equation for $\Ad_U$ in terms of powers of $\Ad_U^{-1}$, we take linear combinations of \Cref{eq:negative_t_right_wall} to derive each of the left-wall conditions in \Cref{eq:left_commutator}.
    The converse statement follows analogously.
\end{proof}
\begin{proof}

% Consider the algebra of all left (right) local operators, $\mathcal{M}_L$ ($\mathcal{M}_R$).
% Following \Cref{def:walls}, 
%The wall condition is most easily understood as defined over the operator algebra on $L$, denoted $\mathcal{M}_L$, and the operator algebra on $R$, denoted $\mathcal{M}_R$ (see \Cref{app:algebras} for details). 
%Since fininte matrix algebras are von-Neumann algebras satisfying the double commutant theorem (see \Cref{app:commutants}), we have $\mathcal{M}_{LC} = \comm(\mathcal{M}_R)$ and $\mathcal{M}_{RC} = \comm({\mathcal{M}_L})$.
%of unitary matrices. Let us construct a projector map $\Pi_{U=\lambda}$ to the subspace where $U$ acts as a uni-modular scalar $\lambda \in \mathds{C}$. This is a slight generalisation of the construction in \cite{Lindblad1999}, where the projector is not defined for a single eigenvalue.
% $\Pi_{U=\lambda} : B(LCR) \rightarrow B(LCR)$ 
We denote by $\mathrm{spec}(\Ad_U)$ the set of eigenvalues of the map $\Ad_U$ and by $\Pi_\lambda$ the corresponding eigen-projectors, so that 
\begin{align}\label{eq:spec decomp}
  \Ad_{U^t} = \!\!\sum_{\lambda\in \mathrm{spec}(\Ad_U)}\!\!\!\! 
  \lambda^t\, \Pi_\lambda\ .
\end{align}
It is proven in \cite{Yosida1965} that these eigen-projectors can be written as
\begin{equation}
    \Pi_{\lambda} = \lim_{N\rightarrow \infty} \frac{1}{N} \sum_{t=0}^{N-1}\lambda^{-t}\mathrm{Ad}_U^t\ .
    \label{eq:projector}
\end{equation}
We will now substitute the form of the adjoint map into the right-wall condition \Cref{eq:right_commutator} in \Cref{lemma:walls_as_commutators}:
%Equipped with the above construction, we proceed by taking the right-wall condition and Fourier-transforming the equation from time to frequency domain (i.e. to $\lambda$):
We now take linear combinations over the left-hand side of the right-wall conditions,
\begin{equation*}
  \lim_{N\rightarrow \infty}\frac{1}{N}\sum_{t=0}^{N-1} \lambda^{-t} [\mathcal{M}_L, \mathrm{Ad}_U^t(\mathcal{M}_R)] = [\mathcal{M}_L, \Pi_{\lambda}(\mathcal{M}_R)],
\end{equation*}
for every $\lambda \in \mathrm{spec}(\Ad_U)$.
This amounts to a Fourier transformation from the time domain to the frequency domain (i.e. to $\lambda$).
Taking the same combinations of the right-hand side establishes that,
\begin{equation}
  [\mathcal{M}_L, \Pi_{\lambda}(\mathcal{M}_R)] =0,\label{eq:pi_commutator}
\end{equation}
for every $\lambda\in\mathrm{spec}(\Ad_U)$.
%since $\Pi_{\lambda}$ maps to disjoint spaces for different values of $\lambda$. The converse statement also holds since if all eigen-projectors of $U$ preserve the localisation condition, so will $U$ itself. 
Analogously, by using \eqref{eq:spec decomp} we can prove that \eqref{eq:pi_commutator} implies \eqref{eq:right_commutator}. So both identities constitute equivalent formulations of the right-wall condition. `

Next, we recall that 
\begin{align}\label{eq:spec decomp U-1}
  \Ad_{U^{-1}} = \!\!\sum_{\lambda\in \mathrm{spec}(\Ad_U)}\!\!\!\! 
  \lambda^{-1}\, \Pi_\lambda\ .
\end{align}
%Notice that the $\lambda$-invariant subspace projected onto is invariant under time-reversal: $\Pi_{U=\lambda} =\Pi_{U^{-1}=\lambda^*}$ since $U$ and $U^{-1}$ share eigenspaces. Thus, one can rewrite the right-wall condition:
Thus, a linear combination of
\eqref{eq:pi_commutator} with coefficients $\lambda^{-t}$ yields
\begin{equation}
  [\mathcal{M}_L, \mathrm{Ad}_{U^{-1}}^t(\mathcal{M}_R)] = 0\ .
\end{equation}
%\begin{equation}
%[\mathcal{M}_L, \Pi_{U^{-1}=\lambda}(\mathcal{M}_R)] = 0 \text{ for all } \lambda.
%\label{eq:left_wall_implication}
%\end{equation}
Finally, using the identity
%By removing the time-evolution from the left-wall condition, we find that the $U^{-1}$ localises operators on $R$:
\begin{equation*}
    \Ad_U^t [\mathcal{M}_L, \Ad_{U^{-1}}^t (\mathcal{M}_R)]
    =[\Ad_{U}^t(\mathcal{M}_L),\mathcal{M}_R] ,
\end{equation*}
we arrive at
%From this, it is evident that 
\Cref{eq:left_commutator} This concludes the proof that \Cref{eq:right_commutator} implies \Cref{eq:left_commutator}. The converse statement follows analogously.
%is equivalent to the left-wall condition.
\end{proof}

\noindent Under the action of $\Ad_U$ for a wall unitary $U$, the operator space splits into commuting sectors (fragments) therefore decoupling the $L$ and $R$ subsystems simultaneously, as shown on \Cref{fig:left_equiv_right}. The wall condition is a strong form of operator localisation in the dynamical evolution which implies stringent algebraic constraints on the unitary's eigenvectors. In particular, a random generic (e.g. Haar-random) unitary would only satisfy these constraints with probability zero. 

However, such a constraint leads to stability against local perturbations: the wall unitary can be embedded into an arbitrary local circuit environment without affecting the decoupling. Walls permit generic local evolution within subsystems $L, R$ which can arise in coupling these to an environment. By interleaving the Floquet unitary $U$ with arbitrary, possibly non-unitary and time-dependent, local transformations in $L$ and in $R$, the wall condition is preserved. Mathematically, consider a sequence of (local) quantum channels $\Phi^{(\tau)}$ acting on $L$ for $\tau = 1, 2,\ldots$ Then, the localisation conditions holds for all $t \ge 1$:
\begin{align}
  \nonumber
  & \left(\Phi^{(t)}_L \circ \Ad_U \right) 
  \circ \cdots \circ
  \left(\Phi^{(1)}_L \circ \Ad_U \right) 
  (\mathcal{M}_L \otimes \Id_{CR}) 
  \\ &\ \subseteq \mathcal{M}_{LC} \otimes \Id_R\ .
\end{align}
From a many-body physics perspective, wall unitaries lead to non-ergodic dynamics which is inherently robust under arbitrarily strong local perturbations in $L$ and $R$. This prevents the phenomenon of avalanche-instability reported in the many-body localisation literature \cite{Morningstar2022, Abanin2019, Ha2023} and will form the basis for comparing dynamics due to walls with Hilbert-space fragmentation \cite{Moudgalya2022, Kovacs2026, Logaric2025} in later sections.

\subsection{Localised operator sub-algebras \label{sec:sub_algebra_embedding}}
In this section, we consider how the evolution of $L$-local and $R$-local algebras overlap in the $C$-system. This will lead us to showing that the the wall unitary acts invariantly on an sub-algebra which will constrain the structure of the unitary.  We will focus on the $L$-local operator evolution from the left-wall condition but under \Cref{thm:left_right_equiv}, all the results will carry over to right-evolving operators.

 \begin{definition}
     Let $\overline{\mathcal{M}_L}$ denote the algebraic closure of operators obtained by evolving local operators in $L$ to arbitrary times:
     \begin{align}
            \overline{\mathcal{M}_L} &= 
            \mathrm{Alg}{\left \{  \Ad_U^t(\mathcal{M}_L\otimes \Id_{CR})  \right \}_{t=0}^{\infty}}.
     \end{align}
     Analogously, for time-evolved operators from $R$-local operators, we define: 
     \begin{equation}
        \overline{\mathcal{M}_R} = 
            \mathrm{Alg}{\left \{\Ad_U^t(\Id_{LC} \otimes \mathcal{M}_R)   \right \}_{t=0}^{\infty}}
    \end{equation}
 \end{definition}
\noindent Note that the $\overline{\mathcal{M}_{L, R}}$ remain finite algebras despite the $t \rightarrow\infty$ limit. In the following, we will refer to these as \textit{embedded sub-algebras}.

\begin{corollary}
    $\Ad_U$ acts invariantly on $\overline{\mathcal{M}_L}$ and $ \overline{\mathcal{M}_R}$.
    \label{cor:invariant_algebras}
\end{corollary}

\begin{lemma}
    If $U$ is a wall unitary then $\overline{\mathcal{M}_L}$ and $\overline{\mathcal{M}_R}$ commute.
    %, $ \left [ \overline{\mathcal{M}_L}, \overline{\mathcal{M}_R} \right] = 0$.
\end{lemma}
\begin{proof}
    For any $t_1,t_2 \ge 0$, from either the left-wall or right-wall condition, we use the invariance of the commutator under the time evolution to obtain
    $\left [\Ad_U^{t_1} (\mathcal{M}_L), \Ad_U^{t_2} (\mathcal{M}_R) \right] = 0$.
    Hence, the generators of $\overline{\mathcal{M}_L}$ and $\overline{\mathcal{M}_R}$ pairwise commute and therefore the algebras as a whole commute.
\end{proof}

We interpret these algebras as the minimal algebras that include the time-evolved $L$-local and $R$-local operators as we took the deformations of local operator algebras as the generators. We conjecture that every element of $\overline{\mathcal{M}_L}$ can be created starting from a local operator in $L$ through evolution by $\Ad_U$ and a time-dependent $L$-local channel in finite time (and similarly for $\overline{\mathcal{M}_R}$). Note that $\overline{\mathcal{M}_L}$ includes $C$-local operators that correspond to the local conserved quantities of the $\Ad_U$ map.

We now derive the structure of these algebras.

\begin{theorem}
  If $U$ is a wall unitary then there are two subalgebras $\mathcal{A}_C, \mathcal{B}_C \subseteq \mathcal{M}_C$ such that 
  \begin{align}\label{eq:barML}
        \overline{\mathcal{M}_L}  = \mathcal{M}_L \otimes \mathcal{A}_C \otimes \Id_R, \\
        \overline{\mathcal{M}_R}  = \Id_L \otimes\mathcal{B}_C \otimes \mathcal{M}_R.
  \end{align}
  The sub-algebras pairwise commute, that is,
%  where $\mathcal{A}, \mathcal{B} \subseteq \mathcal{M}_C$ are mutually commuting sub-algebras: 
    \begin{align}
        &\mathcal{A}_C \subseteq \mathrm{Comm}_C(\mathcal{B}), \\
        &\mathcal{B}_C \subseteq \mathrm{Comm}_C(\mathcal{A}), 
    \end{align}
    \noindent where we have taken the commutant in $\mathcal{M}_C$.
    \label{thm:left_right_closure}
\end{theorem}
\begin{proof}
    From the wall condition, we have $\overline{\mathcal{M}_L} \subseteq \mathcal{M}_{LC}$ and clearly $\mathcal{M}_L \subseteq \overline{\mathcal{M}_L}$.
    We can then apply \Cref{lemma:subalgebra_product} from the Appendix to obtain the decomposition \eqref{eq:barML}.
    The result for $\overline{\mathcal{M}_R}$ is obtained in the same fashion.
    Since $[\overline{\mathcal{M}_L}, \overline{\mathcal{M}_R}] = 0$, it follows that $\mathcal{A}, \mathcal{B}$ pairwise commute.
\end{proof}
The previous theorem establishes that the localisation due to wall unitaries stipulates the existence of commuting matrix sub-algebras which contain the evolution of local operators. These algebras are unique to the wall, however, $\mathcal{A}_C, \mathcal{B}_C$ algebras are not necessarily maximal only if $\mathcal{A}_C = \comm_C(\mathcal{B}_C)$. These will be special cases we comment on later. We note that given the choice of $\overline{\mathcal{M}_{L,R}}$ in the form given in \Cref{thm:left_right_closure}, a unitary satisfying the invariance condition in \Cref{cor:invariant_algebras} will be a wall hence the invariance of an embedded algebra is an equivalent characterisation of the wall. This, however, is not a sufficiently tight constraint to give an explicit form of the unitary (a unique parametrisation). Instead, the wall condition will imply a block-diagonal structure as we show in \Cref{sec:rep_theory}.

We now turn to the edge cases where the embedded sub-algebra is not proper and show that this is trivial localisation due to the periodic repetition of a product gate. 

\begin{theorem}
    The wall unitary is of bi-partite product form: $U = V_L \otimes W_{CR}$ or $U = V_{LC} \otimes W_R$ if and only if the sub-algebra $\mathcal{A}_C$ is improper under \Cref{thm:left_right_closure}, that is, $\mathcal{A}_C = \langle \Id \rangle $ or $\mathcal{A}_C = \mathcal{M}_C$.
    \label{thm:0-walls}
\end{theorem}
\begin{proof}
    By straightforward application of $\Ad_U$ on left/right local operators, any bi-partite product unitary satisfies the wall condition without any additional algebraic restriction. We focus on the other side of the implication.
    If $\mathcal{A}_C =  \mathcal{M}_C $, then $\comm(\mathcal{A}_C) = \langle \Id \rangle$. From $\langle \Id \rangle \subseteq \mathcal
    B$ it follows that $\mathcal{B} = \langle \Id \rangle$. Then, since $\Ad_U^t(\mathcal{M}_R) \subseteq \mathcal{M}_R$, the wall acts locally on $R$ therefore $U$ is factorisable as a local unitary.
    For the case, $\mathcal{A}_C = \langle \Id \rangle$, a wall extends from $L$ to $RC$ and under \ref{thm:left_right_equiv}, this implies that there is a bounded light cone from $CR$ to $L$ also. Then, a similar reasoning establishes that the unitary is product across $L$ and $CR$. 
\end{proof}
\noindent \Cref{thm:0-walls} is equivalent to Lemma IV.1 of \cite{Kovacs2026}. There, a different proof technique was utilised based on the direct evaluation of the commutator norm as an operator entropy measure. We presented here an alternative proof based on the algebraic characterisation. The theorem implies that any non-trivial decoupling of the system involves the embedding of a proper sub-algebra $\mathcal{A}_C\subset \mathcal{M}_C$. This should be understood as a constraint on the maximum attainable Schmidt rank of operators, as they evolve under the unitary map, as we shall see in \Cref{sec:entanglement}. Unless otherwise stated in the following, we will focus on walls which are not reducible to a product unitary.
\section{Local conserved quantities}

\begin{figure}
    \centering
    \includegraphics[width=\linewidth]{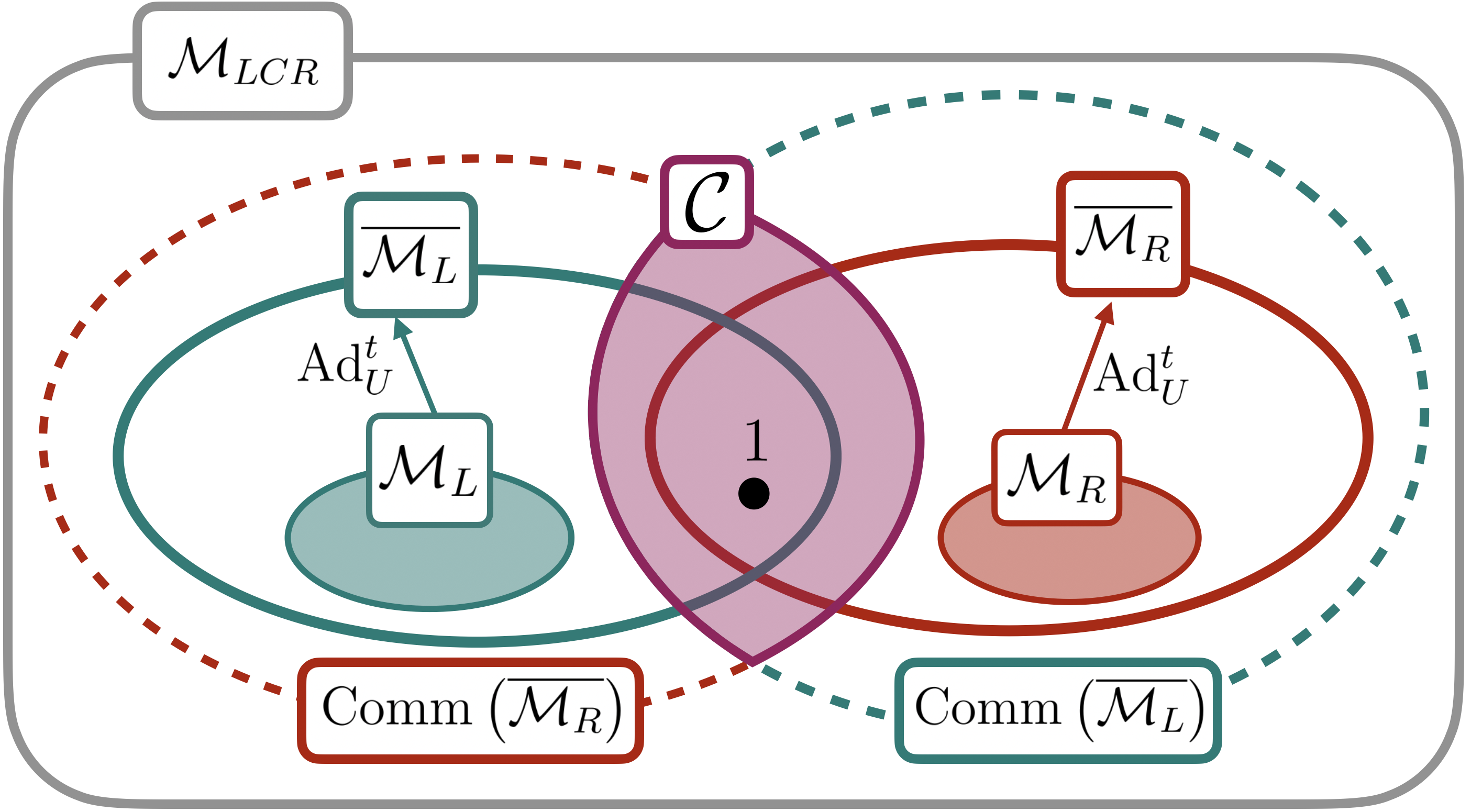}
    \caption{ \small Algebraic structure of conserved operators. A bounded light cone splits operator space with the algebra $\mathcal{C}$ of $C$-local conserved operators at the intersection of commutants from left-localised and right-localised operators. Local conserved charges can  be independent from left-right coupling (ie. uncoupled subsystems) which are not elements of the localised algebras. } 
    \label{fig:venn_diagram}
\end{figure}

 We have seen that the wall condition is equivalent to the invariance of commuting sub-algebras overlapping on $C$. In this section, we ask under what condition is this accompanied by a conservation law: the existence of $C$-local operators which evolve locally in the $C$ subsystem.  We consider the abstract algebraic characterisation and move on to representation theory in the later sections. 

\begin{definition}[Local conserved charge]
    Let $x\in \mathcal{M}_C$. We say that $x$ is a local conserved charge if $\mathrm{Ad}^t_U(x) \in \mathcal{M}_C$ for all $t\geq0$.
    Clearly, the local conserved charges form an algebra $\mathcal{C} \subseteq \mathcal{M}_C$.
\end{definition}

\noindent The above definition includes operators which commute with the unitary but local conserved charges are more generic, i.e. they include central operators which evolve up to an arbitrary phase, ie. $c$ such that $[c, U] = \exp(i\theta_c)c$ for some real scalar $\theta_c$. 
There is, however, a superoperator which will be left invariant, as opposed to up to a phase.

\begin{lemma}
    A projection superoperator to $\mathcal{C}$ the algebra of conserved charges is a commuting superoperator.
\end{lemma}
\begin{proof}
    Let $P_{\mathcal{C}}$ be the orthogonal (according to the Hilbert-Schmidt inner product) projection superoperator onto $\mathcal{C}$.
    Let $c \in \mathcal{C}$ and $d \in \mathcal{C}_\perp$, the orthogonal complement.
    Since $P_\mathcal{C} c = c$, we have $\Ad_U P_\mathcal{C} c = \Ad_U c = \Ad_U P_\mathcal{C} c$.
    From the same logic as \Cref{thm:left_right_equiv}, $\mathcal{C}$ is closed under $\Ad_U^{-1}$ and hence $\mathcal{C}_\perp$ is closed under $\Ad_U$.
    Since $P_\mathcal{C} d = 0$, we have $\Ad_U P_\mathcal{C} d = 0 = P_\mathcal{C} \Ad_U d$.
    Since $\mathcal{C} \oplus \mathcal{C}_\perp$ is the whole space, $\Ad_U$ and $P_\mathcal{C}$ commute.
\end{proof}
\noindent In the following, we can treat conserved charge and $U$-invariant local sub-algebras on the same footing. The algebra $\mathcal{C}$ can be thought of as a symmetry of the adjoint super-operator which has been discussed in the literature as a model for non-universal dynamics \cite{Lastres2026}. Super-operator symmetries imply an invariant operator subspace without needing any particular operator to be conserved.

\begin{theorem}
    The algebra of local conserved charges on $C$ is the intersection of commutants: \begin{equation}\mathcal
    C = \comm\left ({\overline{\mathcal{M}_L}} \right) \cap \comm\left ({\overline{\mathcal{M}_R}} \right).\end{equation}
    \label{thm:central_conserved_charge}
\end{theorem}
\begin{proof}
    First, we note that $\comm(\mathcal{M}_L) = \mathcal{M}_{CR}$ and $\comm(\mathcal{M}_R) = \mathcal{M}_{LC}$, whence $\comm(\mathcal{M}_L) \cap \comm(\mathcal{M}_R) = \mathcal{M}_C$.
    From the construction of $\overline{\mathcal{M}_L}$, we have $\mathcal{M}_L \subseteq \overline{\mathcal{M}_L}$ and similar for subsystem R.
    The inclusion ordering is reversed by taking commutants, thus $\comm(\overline{\mathcal{M}_L}) \subseteq \comm(\mathcal{M}_L)$.
    Therefore, $\comm(\overline{\mathcal{M}_L}) \cap \comm(\overline{\mathcal{M}_R}) \subseteq \mathcal{M}_C$, i.e. it is local to subsystem C.

    Now we turn to showing that it is conserved.
    Let $x \in \comm(\overline{\mathcal{M}_L}) \cap \comm(\overline{\mathcal{M}_R})$ and $y = \Ad_U(x)$.
    Since $[x, \overline{\mathcal{M}_L}] = 0$, we have $[y, \Ad_U \left(\overline{\mathcal{M}_L}\right)] = 0$.
    But $\overline{\mathcal{M}_L}$ is invariant under $\Ad_U$.
    Hence, $y \in \comm(\overline{\mathcal{M}_L})$ and by a similar argument must also be in $\comm(\overline{\mathcal{M}_R})$.
    Therefore, $\comm(\overline{\mathcal{M}_L}) \cap \comm(\overline{\mathcal{M}_R})$ consists of local conserved charges on C and is thus contained in $\mathcal{C}$.

    Conversely, since $\mathcal{C}$ is invariant and local to C we have $[\Ad_U^t (\mathcal{C}), \mathcal{M}_L] = 0$ for all $t \ge 0$.
    We rewrite this as $[\mathcal{C}, \Ad_U^{-t} (\mathcal{M}_L)] = 0$ for all $t \ge 0$ and follow the same logic as in \Cref{thm:left_right_equiv} to find $[\mathcal{C}, \overline{\mathcal{M}_L}] = 0$. The same can be done for $\mathcal{M}_R$.
    Hence $\mathcal{C}$ is contained in $\comm(\overline{\mathcal{M}_L}) \cap \comm(\overline{\mathcal{M}_R})$.
    Therefore, the two algebras are in fact equal.
\end{proof}
\begin{corollary}
    $\langle \Id \rangle \subseteq \mathcal{C}$. Non-trivial intersection imply non-trivial conserved quantities under \Cref{thm:central_conserved_charge}. 
\end{corollary}
\begin{corollary}[Commutative algebras are conserved]
    If $\mathcal{A}$ or $\mathcal{B}$ is Abelian then every element is a conserved charge.
\end{corollary}
\begin{proof}
    Abelian algebras are included in their commutants by definition, that is,  $\mathcal{A} \subseteq\comm_C(\mathcal{A})$. Since $\mathcal{A} \subseteq \comm_C(\mathcal{B})$, we have that $\mathcal{A} \subseteq \mathcal{C}= \comm_C(\mathcal{A}) \cap \comm_C(\mathcal{B})$ as desired. Same reasoning applies to commutative $\mathcal{B}$.
\end{proof}

\noindent The previous theorem establishes the general condition for existing conserved quantities as illustrated on \Cref{fig:venn_diagram}. From this splitting of operator space, there can be conserved quantities which don't participate in the dynamics of the $L-R$ coupling and therefore are conserved (these are independent subsystems in the circuit). For operators which do participate in the coupling, the conserved quantities are determined from the intersection of Abelian sub-algebras of $\mathcal{A}, \mathcal{B}$ with the commutants.

 For non-abelian embedded algebras, the wall can be constructed so that not every element of the algebra is conserved. Under \Cref{thm:left_right_closure}, if one chooses the right-localised algebra $\mathcal{B}_C$ to be equal to $\comm_C(\mathcal{A}_C)$, then the conserved charges are the centre of the embedded algebra $\mathcal{C} = \comm_C(\mathcal{A}_C)\cap \comm_C(\mathcal{B}_C) =  Z(\mathcal{A}_C)$ (see \Cref{app:commutants}). However, a finite non-abelian algebra is always the direct sum of full matrix blocks (see \Cref{app:rep_theory}) which have trivial centre spanning only the identity \cite{Bény2020}. In these cases, the conserved operators are the Abelian algebra generated by the projectors to the irreducible matrix blocks which is strictly lower dimensional than the full matrix block. 

Local operators translated in time (but not in space) up-to a phase are sometimes referred to as still-\enquote*{solitons} in the literature which lead to non-ergodic evolution in dual-unitary circuit dynamics \cite{Bertini2020_op_ent1, Bertini2020_op_ent2, HoldenDye2025fundamentalcharges}. Some wall instances will feature these solitons (up to the gauge freedom discussed in the previous section) with an Abelian embedded algebra. However, the bounded light cone constraint is more general in that there can be instances where there is central-local operator space fragment which is decoupled from the left/right subsystems and therefore do not have a local charge. In these instances, the operator space is split fragments that commute at the central subsystem but are not separable both at $L-C$ and $R-C$ boundary. In these instances, there is a \enquote*{hidden} subsystem which will link these instances to error-correcing codes. We will see an example of this in the following section. 

\section{Time-dependent walls from gauge freedom}

Not only time-periodic evolution can exhibit causal independence in the sense of \Cref{def:walls}. This will lead to an important restriction of our results due to gauge freedom, namely that the wall is only defined up to a global unitary isomorphism class.

\begin{lemma}[Wall sequence]
    Let $(U_{\tau})_{\tau = 1}^{\infty}$ be a sequence of unitaries where the algebra $\overline{\mathcal{M}_L}$ is left invariant at all times:
    \begin{equation}
        \Ad_{U_{\tau}} \left ( \overline{\mathcal{M}_L} \right) = \overline{\mathcal{M}_L} \text{ for all } \tau \ge 1.
    \end{equation}
    Then, the wall conditions hold for all $\tau \ge 1$ under the evolution induced by the sequence. We then refer to $(U_{\tau})_{\tau = 1}^{\infty}$ as a wall sequence.
    \label{lemma:wall_sequence}
\end{lemma}

\begin{proof}
    Follows from \Cref{cor:invariant_algebras}. Note that that the invariance of $\overline{\mathcal{M}_R}$ follows automatically since commutation relations are preserved under any unitary isomorphism.
\end{proof}
The wall condition does not prevent a restricted notion of ergodicity within the localised subspaces. We may have wall sequences which have generic dynamics such as uniform mixing of operators. In this sense, it is an example of Hilbert-subspace ergodicity recently discussed in literature \cite{Logaric2025}. Under wall dynamics the subspaces are spatially localised subsystems.

\begin{figure}
    \centering
    \includegraphics[width=\linewidth]{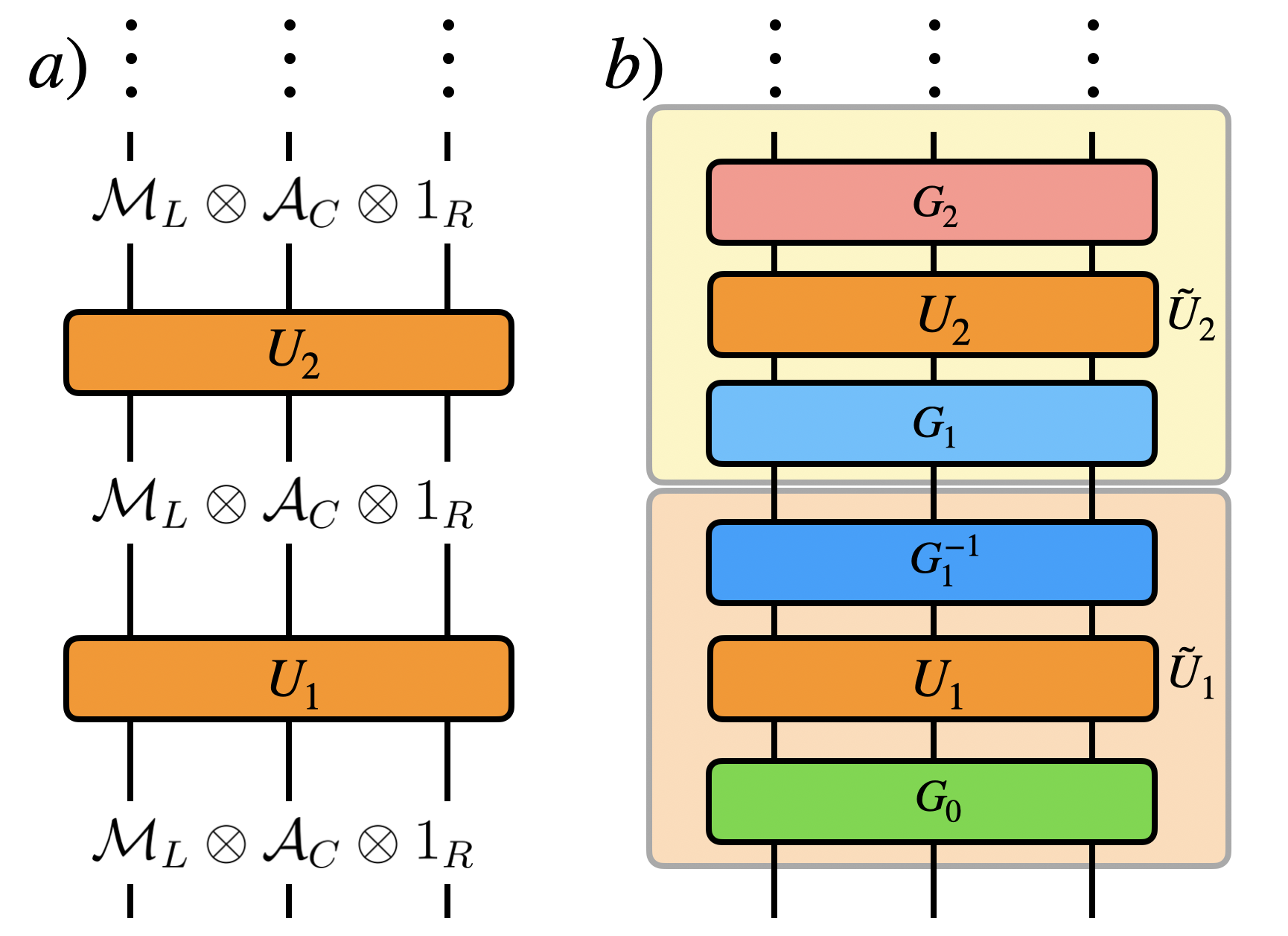}
    \caption{ \small Extending the wall condition to time-dependent unitary sequences. On \textit{a)}, the sequence consistes of elements of the normaliser (see \Cref{sec:rep_theory}) group of $\overline{\mathcal{M}_L}$. On \textit{b)}, we dress this sequence through gauge transformations, inducing a sequence of isomorphic sub-algebra transformations, changing the global operator basis at every timestep.}
    \label{fig:time_dep_extension}
\end{figure}

\begin{definition}[Gauged wall sequence]
    Let $(U_{\tau})_{\tau = 1}^{\infty}$ be a wall sequence under \Cref{lemma:wall_sequence}. We define a gauged wall sequence as $(\Tilde{U}_{\tau})_{\tau = 1}^{\infty}$ under suitably chosen gauge transformations $G_0, G_1, ... ,G_{\tau}$:
    \begin{equation}
        \Tilde{U}_{\tau} = \prod_{t=1}^{\tau}{G}^{-1}_tU_t{G_{t-1}},
    \end{equation}
    where $\Ad_{G_0}$ acts invariantly on $\overline{\mathcal{M}_L}$ and $G_{t}$ are arbitrary unitaries for $t>1$. 
    \label{def:gauged_wall_sequence}
\end{definition}

From the previous definition, a gauged wall sequence violates the wall conditions at timesteps $\tau = 1, 2, ...$ since the gauge transformations are not required to preserve localisation. However, due to the underlying wall-sequence, the evolution merely induces a sequence of isomorphic transformations: 

\begin{equation}
    \mathcal{L}_0 = \overline{\mathcal{M}_L} \overset{\Ad_{\Tilde{U}_1}}{\longrightarrow} \mathcal{L}_1 \overset{\Ad_{\Tilde{U}_2}}{\longrightarrow} \mathcal{L}_2 \overset{\Ad_{\Tilde{U}_3}}{\longrightarrow} ...,
\end{equation}
\noindent where $\mathcal{L}_i \cong \mathcal{L}_j$ for all $i, j \geq 0$, therefore the phenomenology of localisation still holds in a non-local operator basis. Another intepretation is that the localised operator evolution is \enquote*{obfuscated}, so that the observed operator evolution appears random to mask the invariance of a localised subspace. 

The gauged construction highlights that the wall condition is only well-defined up to the class of isomorphisms of the embedded sub-algebras due to the arbitrary choice of global operator basis in the circuit. The sequence above need not terminate as there are infinitely many isomorphic sub-algebras. If the sequence recurs, then the gauged sequence reduces to a a time-periodic wall sequence under coarse-graining time. In the following theorem, we formalise this intuition.
\begin{theorem}
    Consider the sequence of isomorphic algebras $(\mathcal{L}_{\tau})_{\tau=1}^{\infty}$ under the evolution of a gauged wall sequence $(\Tilde{U}_{\tau})_{\tau = 1}^{\infty}$. If the sequence recurs, the evolution is generated by a time-periodic wall sequence acting invariantly on a static sub-algebra. That is, there exists gauge transformations $G, H$ so that:
    \begin{equation}
        \Tilde{V}_{\tau} = \left( G \left (\prod_{t=1}^{T-1}    \Tilde{U}_t\right)H \right)^{\tau}.
    \end{equation}
    $(\Tilde{V}_{\tau})_{\tau=1}^{\infty}$ is then a time-periodic wall sequence which leaves the static algebra $\mathcal{L}$ invariant at every step. $T$ is the finite recurrence time (or equivalently the size) of the sequence of algebras.
\end{theorem}
\begin{proof}
If sequence returns to itself in finite $T$, then there exists $ P\leq T $, so that:
    \begin{equation}
        \Ad_{\Tilde{U}_{T}}(\mathcal{L}_T) = \mathcal{L}_P.
    \end{equation}
The unitary evolution is reversible, so we have: 
\begin{align}
    \Ad_{\Tilde{U}^{-1}_{P}}(\mathcal{L}_P) &= \mathcal{L}_{P-1}  = \mathcal{L}_T.
\end{align}
By iteration of the reverse evolution, the algebras up to $P-1$ included already between indices $P$ to $T$. Therefore, the sequence reduces to a single loop and we take $\mathcal{L}_{T+1} = \mathcal{L}_1$ without loss of generality.

In this case, we choose the gauge transformation according to:
\begin{align}
    \Ad_H (\mathcal{L}) &= \mathcal{L}_1, \\
    \Ad_G (\mathcal{L}_T) &= \mathcal{L}, 
\end{align}

\noindent where we have introduced an arbitrarily chosen isomorphic sub-algebra $\mathcal{L} \cong \mathcal{L}_i$. Then, by substitution, the transformed, time-periodic, sequence $\Tilde{V}_{\tau} = \left( G \left (\prod_{t=1}^{T-1}    \Tilde{U}_t\right)H \right)^{\tau}$ leaves the sub-algebra $\mathcal{L}$ invariant at every step. 
\end{proof}

We expect that the converse of the above statement does not hold, since periodic repetitions can give rise to aperiodic sequences in the algebras which will not lead to a strict recurrence. One should view the above statement as an analogue of Floquet's theorem showing that static invariant sub-algebras are always associated with time-periodic evolution \cite{Floquet1883, Shirley1965}. An interesting outstanding question is how to construct continuous time-evolution (either periodic or non-periodic) to preserve the wall structure. This would involve understanding the relation between Lie algebras in Hamiltonian evolution and the desired $\mathrm{C}^*$-algebraic invariant structure  that we leave for future work \cite{Humphreys1972IntroLieAlgebras, MathieuVillena2003LieDerivations, BresarKissinShulman2008LieIdeals}.

\section{Structure of brickwork walls \label{sec:rep_theory}}
\subsection{Representation theory}
\begin{figure}
    \centering
    \includegraphics[width=\linewidth]{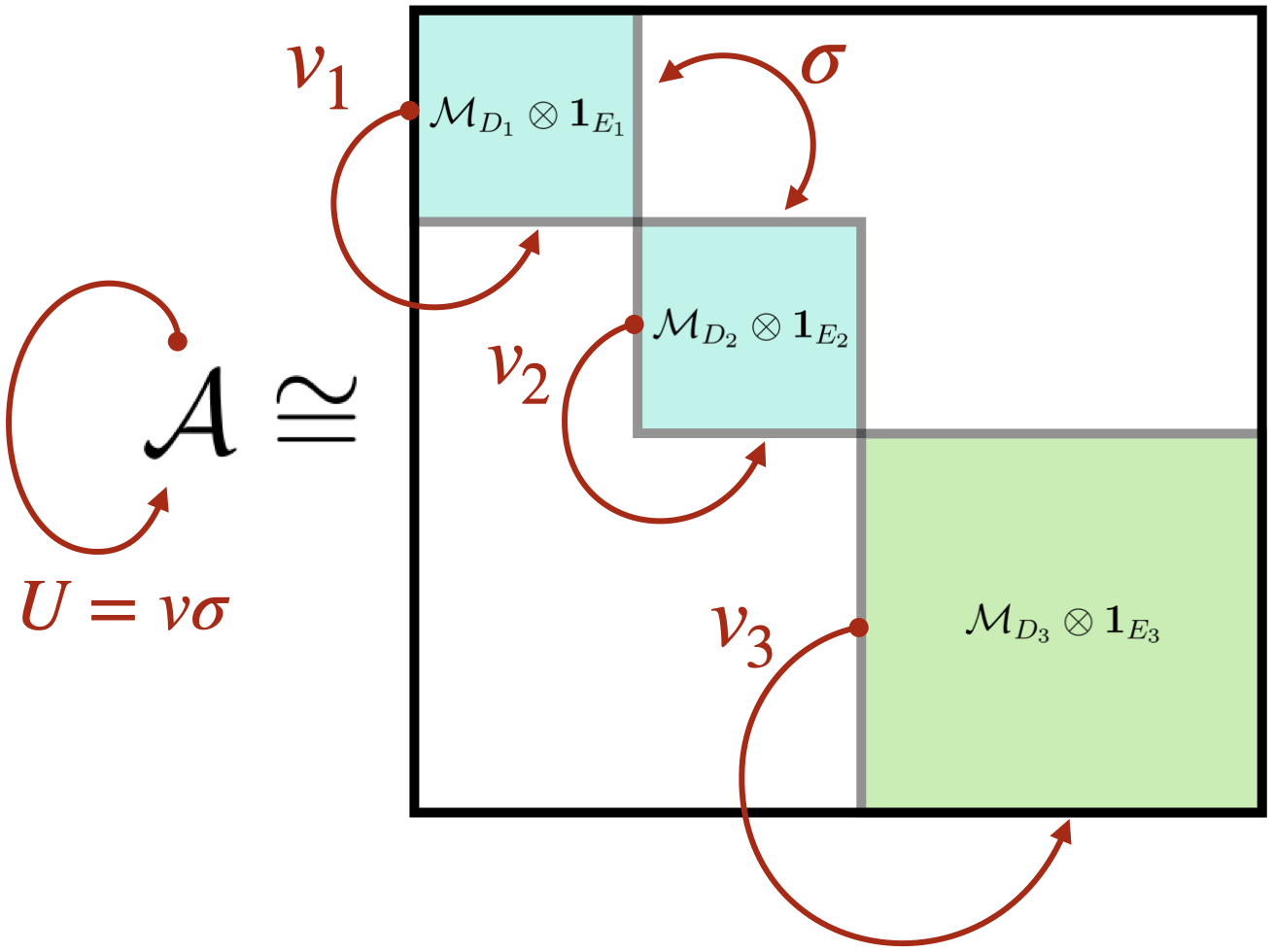}
    \caption{The structure of unitary automorphisms of an operator algebra $\mathcal{A}$. $\mathcal{A}$ decomposes as a direct sum of full matrix blocks, some of which are equivalent in terms of the subspace dimensions $D_i$, $E_i$.
    Elements of the automorphism group $\mathrm{U}(\mathcal{A})$ decompose as $U=v\sigma$ with $v$ acting within each irreducible block as $v_i$ and between the blocks by a permutation $\sigma$ of equivalent (but independent) subspaces.}
    \label{fig:automorphism_struct}
\end{figure}
In this section, we utilise the representation theory of finite $\mathrm{C}^*$-algebras in order to construct the detailed structure of wall unitaries in brickwork circuits and that of local conserved charges.  We will focus on representations of a left-wall for convenience  and refer to formal literature on representations of finite algebras in \Cref{app:rep_theory} while this section focuses on their application. From \Cref{cor:invariant_algebras}, a wall unitary leaves an embedded sub-algebra invariant. This implies that the wall is a representation of a unitary automorphism of the sub-algebra. While this invariance property uniquely identifies that $U$ is a wall, it does not uniquely specify the unitary's parameters because walls associated with a particular algebra form a unitary sub-group. In this section, we will derive the properties of this group which in turn will specify the structure of all wall unitaries with a particular algebra.

The invariance of the embedded algebra induces a splitting of Hilbert space under a global unitary isomorphism (a choice of operator basis): 
\begin{equation}
    \mathcal{H}_{LCR} \cong \bigoplus_{i=1}^{n_i} = \mathcal
    {H}_L\otimes\mathcal{H}_{
    \mathcal{D}_i} \otimes \mathcal{H}_{\mathcal{E}_i} \otimes \mathcal{H}_R.
\end{equation}
where $\mathcal{D}_i$, $\mathcal{E}_i$ are subspaces of $\mathcal{H}_C$ hosting irreducible matrix blocks of $\mathcal{A}_C$. Let us denote $\dim \mathcal{D}_i = D_i$ and $\dim \mathcal{E}_i = E_i$ as the subspace dimensions. $n_i$ denotes the number of irreducible blocks. Note the dimensional constraint $\sum_{i=1}^{n_i} {D}_i{E}_i = \dim \mathcal{H}_C$ .

Then, we can represent the algebra as full matrix algebras within irreducible blocks according to:
\begin{equation}
    \mathcal{A}_C \cong \bigoplus_{i=1}^{n_i} \mathcal{M}_{\mathcal{D}_i} \otimes \Id_{\mathcal{E}_i}.
    \label{thm:block_structure}
\end{equation}

\noindent The algebra's dimension satisfies $\dim \mathcal{A}_C = \sum_{i=1}^{n_i}( {D}_i)^2$. The commutant algebra follows the same block structure:
\begin{equation}
    \comm(\mathcal{A}_C) \cong \bigoplus_{i=1}^{n_i} \Id _{\mathcal{D}_i} \otimes \mathcal{M}_{\mathcal{E}_i}.
\end{equation}

As we show in \Cref{app:normaliser}, operators that leave an algebra invariant form an automorphism group $\mathrm{U}\left ({{\mathcal{A}_C}} \right)$ that we call the \textit{normaliser} following the terminology from quantum error correction \cite{kribs2006operatorquantumerrorcorrection}. These are unitary maps that leave a code-space invariant, ie. they are logical operations of the encoded information. The group elements act as arbitrary unitaries in the irreducible blocks of $\mathcal{A}_C$ up to permutations of equivalent matrix blocks as on \Cref{fig:automorphism_struct}. Two subspaces are equivalent if both their algebraic dimension $D_i$ and degeneracy space dimension $E_i$ are equal. In \Cref{app:normaliser}, we formalise the preceding intuitive statements to prove the following general theorem, which may be of independent interest.

\begin{theorem}[Structure theorem for the unitary automorphisms]
Up to a global isomorphism, the normaliser of a finite matrix sub-algebra $\mathcal{A}\subseteq\mathcal{M}$ decomposes as:
   \begin{align}
    \mathrm{U}(\mathcal{A}) \cong \left( \bigoplus_{i \in I} \left(\mathrm{U}_{\mathcal{D}_i} \otimes \mathrm{U}_{\mathcal{E}_i} \right)\right) \rtimes \Sigma
  \end{align}
  where $\Sigma$ is the automorphism group for the factors of $\mathcal{A}$ (see \Cref{app:normaliser}). The semidirect product is defined by the right-action,
  \begin{align}
    \left(\bigoplus_{i \in I} v_i \otimes v_i\right) \triangleleft \sigma = \sum_{i \in I} v_{\sigma^{-1}(i)} \times v_{\sigma^{-1}(i)}
  \end{align}
  which simply rearranges equivalent representations of $\mathcal{A}$.
  \label{thm:wall_structure}
\end{theorem}
\begin{proof}
    See \Cref{app:normaliser}.
\end{proof}

From this, the structure of wall unitaries follow straightforwardly. The $L, R$ subsystems evolve without constraint hence we need only extend the irreducible blocks to the left and the right subsystems, since the automorphism group a full matrix algebra $\mathcal{M}_L$, $\mathcal{M}_R$ is any unitary on $L, R$ respectively under the Skolem-Noether theorem (see \Cref{app:rep_theory}).

\begin{corollary}
    Consier a brickwork wall unitary $U = V_{LC}W_{CR}$ composed of bi-partite gates. Then, there is a $C$-local gauge transformation $G$, so that $U = \Tilde{V} \Tilde{W} = (VG^{-1})(GW)$ with the following decompositions: 
    \begin{align}
        \Tilde{V} &= \bigoplus_{i\in I} \tilde T^{\Pi(i)}_{L \mathcal D_{i}} \otimes r^{\Pi(i)}_{\mathcal E_{i}},\\
        \Tilde{W} &= \bigoplus_{i \in I} t^{\pi(i)}_{\mathcal D_{i}} \otimes \tilde R^{\pi(i)}_{\mathcal E_{i} R},
    \end{align} 
    where $t^i, r^i$ are arbitrary unitaries within the irreducible subspaces satisfying $T^i = \Tilde{T}^it^i$ and $R^i = \Tilde{R}^i r^i$. $\Pi, \pi$ are arbitrary permutations of equivalent block labels.
\end{corollary}
\begin{proof}
    Since the wall is only well-defined up to global isomorphisms under \Cref{def:gauged_wall_sequence}, one can carry out an arbitrary gauge transformation in between brickwork layers whilst preserving the splitting of the space (which is a change of local operator basis in between layers). The structure of the constituent bi-partite gates then follows from the more general structure under \Cref{thm:wall_structure}. The $\Tilde{T}, \Tilde{R}$ unitaries act on $L\mathcal{D}_i$ and $\mathcal{E}_iR$ respectively as the $L, R$ only have inner automorphisms (any unitary on the space). The irreducible block permutations of the two gates need not commute.
\end{proof}

In the next section, the explicit construction of wall unitaries will highlight that unitaries acting as block permutations on the algebra don't follow the same block-diagonal form as the decomposed algebra which implies that gates in brickwork walls 
need not commute. As a result, a wall unitary is more generic than unitaries with local symmetries. Note that the above two results do not require the minimality of the region $C$ containing the wall. If $C$ is composed of multiple subsystems, we define minimality of a wall so that one can not enlarge the left and right subsystems by factors of $C$ without breaking the localisation condition. We conclude the section by characterising the algebra of operators that commute with the wall-unitary in terms of the Schmidt vectors of the unitary. 

\begin{lemma}
  The set of operators contained in $C$ that commute with $F$ is
  \begin{align}
      \mathcal Q_C = \bigoplus_i \mathrm{Comm}\mathcal T_{\mathcal D_i} \otimes \mathrm{Comm}\mathcal R_{\mathcal E_i}
  \end{align}
  where $\mathcal T_{\mathcal D_i}$ is the $C^*$-algebra generated by the Schmidt vectors of $T^i_{L \mathcal D_i}$ on $\mathcal D_i$, and $\mathcal R_{\mathcal E_i}$ is the $C^*$-algebra generated by the Schmidt vectors of $R^i_{\mathcal E_i R}$ on $\mathcal E_i$.
\end{lemma}
\begin{proof}
  If $Q_C$ commutes with the Floquet unitary $u$ then $Q_C^\dagger$ commutes with $F$ too. Also, $Q_C$ and $Q_C^\dagger$ commute with $T^i_{L \mathcal D_i} \otimes R^i_{\mathcal E_i R}$ for all $i$, and hence, they also commute with all the elements of the support of $T^i_{L \mathcal D_i} \otimes R^i_{\mathcal E_i R}$ on $\mathcal D_i \otimes\mathcal D_i$, and their Hermitian conjugates.

  In other words, $Q_C$ commutes will all elements of $\bigoplus_i \mathcal T_{\mathcal D_i} \otimes \mathcal R_{\mathcal E_i}$, where $\mathcal T_{\mathcal D_i}$ is the $C^*$-algebra generated by the Schmidt vectors of $T^i_{L \mathcal D_i}$ on $\mathcal D_i$, and $\mathcal R_{\mathcal E_i}$ is the $C^*$-algebra generated by the Schmidt vectors of $R^i_{\mathcal E_i R}$ on $\mathcal E_i$. Equivalently, $Q_C$ belongs to the commutant algebra
  \begin{align}
      \mathcal Q_C = \bigoplus_i \mathrm{Comm}\mathcal T_{\mathcal D_i} \otimes \mathrm{Comm}\mathcal R_{\mathcal E_i}.
  \end{align}
\end{proof}

\subsection{Examples}
\begin{figure*}
    \centering
    \includegraphics[width=0.7\textwidth]{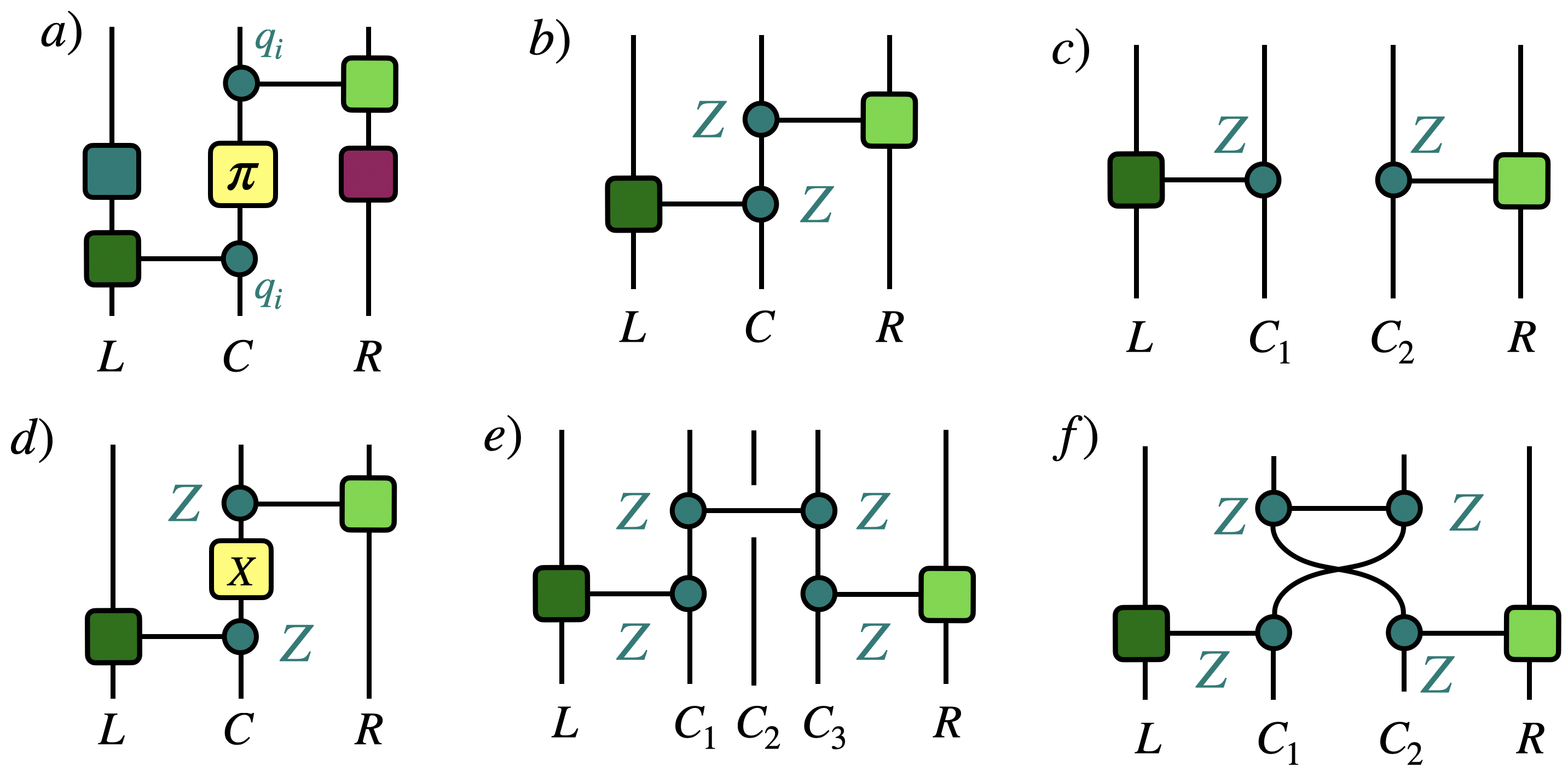}
    \caption{\small Catalogue of Abelian walls. \textit{a)} shows their general structure with two conditional gates on an orthonormal operator basis $\{(q_i)_C = \ket{i}\bra{i}_C \}_i$. These enclose an off-diagonal unitary permutation gate $\pi$ of the basis elements. We use $Z$-conditional gates (see \Cref{app:conditionals}) as a workhorse to illustrate wall properties while we note that this choice of $C$-local operator basis is arbitrary.  \textit{b)} shows a wall with two $Z$-conserving gates. The wall has a local conservation law as $Z_C$ commutes with the wall unitary, and the constituent gates also commute. \textit{c)} shows a reducible wall on a composite central subsystem. While the wall encloses the Abelian algebra $ \langle Z_{C_1} \rangle \otimes \langle Z_{C_2} \rangle$, $L$, $R$ can be extended to form a trivial wall. \textit{d)} features a local soliton on $C$, so that $Z_C \rightarrow -Z_C \rightarrow Z_C$ orbits under the wall dynamics due to the block permutation gate $X$. The constiuent gates don't commute. \textit{e)} has an uncoupled subsystem $C_2$ on which any operator commutes with the wall unitary. This wall is a composition of trivial walls w.r.t. $C_2$ and an Abelian wall on $C_1, C_3$. \textit{e)} features an irreducible wall formed of non-commuting gates. The $\SWAP$ gate acts as a block permutation on the enclosed $Z$-diagonal algebra. The $ZZ$-gate is an inner automorphism of the algebra.}
    \label{fig:abelian_wall}
\end{figure*}

In this section, we consider a few simple representative examples of walls to illustrate the previously outlined mechanism of bounded light-cone formation. 
\subsubsection{Abelian walls \label{sec:abelian_walls}}

Consider the algebra $\mathcal{A}_C = \mathrm{diag(\mathcal{M}_C)}$ as the diagonal operators on $C$ in an arbitrary $C$-local basis, which is a maximal Abelian algebra. In this case, we generate the algebra from pairwise orthogonal projectors $\mathcal{A}_C = \langle \Pi^{(i)}_C\rangle_{i=1}^{\dim C}$ under \Cref{thm:abelian_structure}. The irreducible blocks are all one-dimensional scalars with multiplicity one, ie. $\mathcal{A}_C \cong \mathds{C}^{\dim C}$. The unitary automorphisms of this algebra are diagonal gates with unit-magnitude phases in the eigen-basis any, in addition to any unitary representation of permutations of the diagonal elements which are off-diagonal gates. A generic wall takes the form under \Cref{thm:wall_structure}: 
\begin{equation}
    U = \sum_{a=1}^{\dim C} t^{(a)}_L\otimes \Pi^{\pi(a)}_C \otimes r^{(a)}_R,
\end{equation}
\noindent where $t^{(a)}, r^{(a)}$ are arbitrary unitaries and $\pi \in \mathcal{S}_d$ is an element of the permutation group on $d$ elements. 

We recognise these unitaries are conditional gates (\Cref{app:conditionals}) where the action on the left/right subspaces are determined from the outcome of a projective measurement on the central subsystem. Under \Cref{thm:central_conserved_charge}, all diagonal operators generate a local conserved charge. We may decompose these operators in local circuits as conditional gates from between $LC$ and $CR$ up to basis permutations on $C$, as on \Cref{fig:abelian_wall}. Naturally, any local operator on $L$ or $R$ also leaves this central algebra invariant although we have already considered this equivalence with the definition of the wall. On \Cref{fig:abelian_wall}, we show several examples of walls constructed from $Z$-conditional unitaries (gates which are diagonal in the $C$-local computational basis) to illustrate properties of the wall such as non-commutativity of constituent gates and the notion of reducibility in multipartite systems.

\subsubsection{Non-abelian walls}
\begin{figure}
    \centering
    \includegraphics[width=\linewidth]{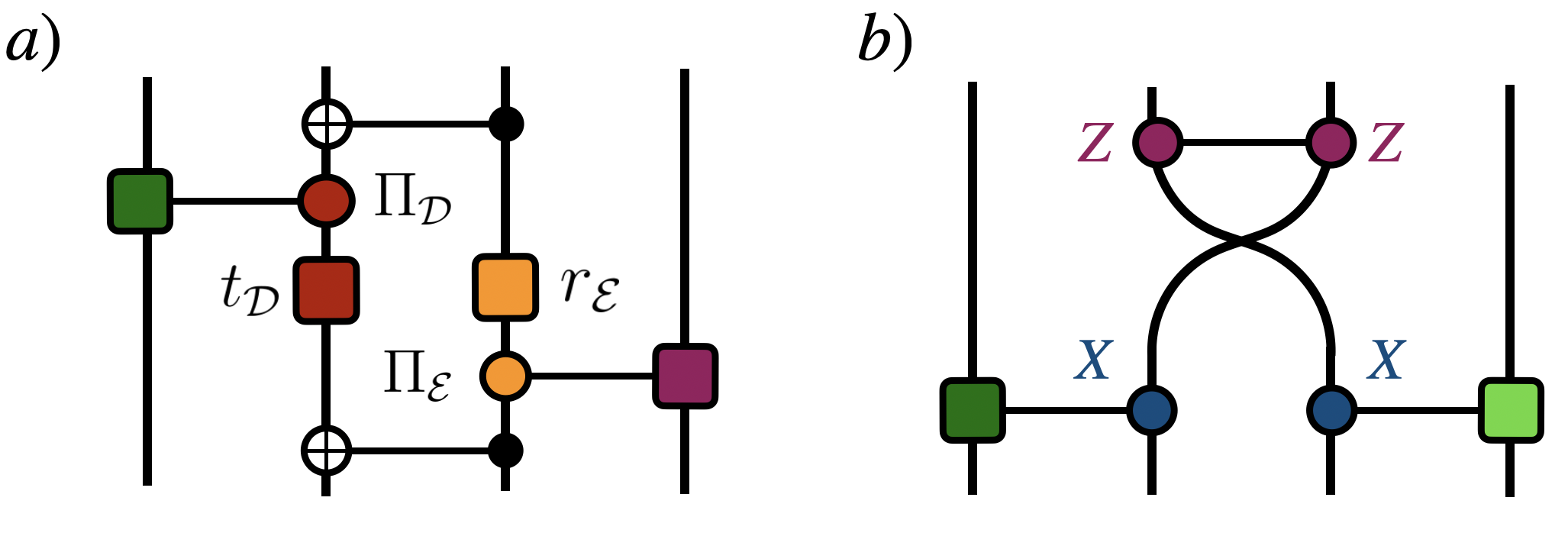}
    \caption{\small \textit{a)} Example of non-abelian wall structure for $\mathcal{A}_C \cong \Id_2 \otimes \mathcal{M}_2$. The embedded algebra has one factor in an entangled basis, hosting an encoded logical susbsystem. \textit{b)} shows a particular realisation by the Clifford $\FSWAP$ gate. 
    Neither gates on \textit{a), b)} host local conservation laws since we chose $\mathcal{B} = \comm(\mathcal{A}_C)$ under \Cref{thm:left_right_closure} so that $\mathcal{A}_C$ has a trivial centre.}
    \label{fig:interfering_wall}
\end{figure}

As a minimal example, we will consider non-commutative algebras on qubit subspaces on $C$. If $\dim C = 2$ (single-qubit), the only non-abelian algebra is the full matrix algebra spanned by two non-commuting Paulis, $\mathcal{M}_C = \langle X_C, Z_C \rangle$. As a result, only trivial non-abelian walls can be constructed under \Cref{thm:0-walls}. On two qubits, consider the non-abelian Pauli algebra:
\begin{equation}
    \mathcal{A}_C = \langle X_{C_1} \otimes \Id_{C_2}, Z_{C_1} \otimes X_{C_2} \rangle.
\end{equation}

\noindent The commutant can be constructed explicitly, by noting that the generators must have $X$ or $\Id$ support on $C_2$ and must commute at even number of positions. Thus we have:

\begin{equation}
    \comm(\mathcal{A}_C) = \langle \Id_{C_2}\otimes X_{C_1}  ,X_{C_2}\otimes Z_{C_1} \rangle.
\end{equation}

By noticing that $\mathcal{A}_C \cap \comm(\mathcal{A}_C) = \langle \Id \rangle$, under \Cref{thm:block_structure}, the algebra is isomorphic to a single factor. The $\mathrm{CNOT}_{21} = \ket{0}\bra{0}_{C_2} \otimes \Id_{C_1} + \ket{1}\bra{1}_{C_2} \otimes X_{C_1}$ induces the isomorphism to block diagonal form:
\begin{align}
    \Ad_{\mathrm{CNOT}_{21}}(\mathcal{A}_C) &= \Id_{\mathcal{E}} \otimes \mathcal{M}_{\mathcal{D}}, \\
    \Ad_{\mathrm{CNOT}_{21}}(\comm(\mathcal{A}_C)) &= \mathcal{M}_{\mathcal{E}} \otimes \Id_{\mathcal{D}}.
\end{align}
\noindent where $\mathcal{E}, \mathcal{D}$ are each two-dimensional irrep subspaces. To construct the unitary automorphisms of this algebra, we take any product unitary $t_{\mathcal{D}} \otimes r_{\mathcal{E}}$ (these parametrise all inner automorphisms multiplied by a commuting unitary in the irrep basis).

From this, the general form of the the non-abelian wall is: 
\begin{equation}
    U = \Ad_{\mathrm{CNOT}_{21}}( T_{L\mathcal{D}} \otimes R_{\mathcal{E}R}).
\end{equation}

The subspaces $\mathcal{D}, \mathcal{E}$ have two-dimensional projectors $\Pi_{\mathcal{D}},  \Pi_{\mathcal{E}}$ thus we can decompose the $L\mathcal{D}$ and $\mathcal{E}R$ couplings as generalised conditional gates:

\begin{align}
    T_{L\mathcal{D}} &= V_L \otimes t_{\mathcal{D}}\Pi_{\mathcal{D}}, \\
    R_{{\mathcal{E}R}} &= r_{\mathcal{E}} \Pi_{\mathcal{E}} \otimes W_R, 
\end{align}

\noindent where $V, W, t, r$ are arbitrary unitaries. We show this structure on \Cref{fig:interfering_wall} along with a specific Clifford-gate example constructed. For a generic (e.g. Haar-random) $V, Q, t, r$, we have $\mathcal{B} = \comm(\mathcal{A}_C)$ so that the $C$-local conserved algebra from \Cref{thm:central_conserved_charge} is simply the centre of $\mathcal{A}_C$: $\mathcal{C} = Z(\mathcal{A}_C) = \mathcal{A}_C \cap \comm(\mathcal{A_C})$. Our particular example has only a single factor which has a trivial centre therefore these walls don't host conserved charges. While one would naively expect the (Abelian) algebra spanned by the projectors, $\langle \Pi_{\mathcal{D}}, \Pi_{\mathcal{E}} \rangle$, to be conserved, the conditional gates ensure that the projectors spread one-sidedly in the wall evolution. In this example, we have utilised an operator algebra well-known from early works of subsystem error-correction \cite{Poulin2005_subsystem_QEC,kribs2006operatorquantumerrorcorrection, Dauphinais2024stabilizerformalism}, where the central subsystem hosts an encoded logical qubit. The condition for causal decoupling is weaker than what is required for error-correction, i.e. walls don't presume recoverability of logical information under errors, but the existence of a code-space and its invariance links the two problems together. 

A particular (Clifford-gate) example of a wall in this category is also shown on \Cref{fig:interfering_wall} where we utilised  the $\FSWAP$ gate as an automorphism of $\mathcal{A}_C$ composed of a controlled-$Z$ gate and a $\mathrm{SWAP}$. $\Ad_{\FSWAP}$ acts (up to a phase) to exchange $X \otimes \Id \leftrightarrow Z \otimes X$ and $\Id \otimes X \leftrightarrow X \otimes Z$. This wall has a salient interpretation through localising Jordan-Wigner fermions. The controlled-$X$ gate \enquote*{injects} a fermion head to the central region which is transported by the $\FSWAP$ gate so that its head commutes with the other controlled-$X$ gate stopping its spreading. Central local operators spread either one-sidedly (e.g. $X_{C_1} \otimes Y_{C_2}$ spreads to $R$ only under the wall unitary while it remains invariant under $\FSWAP$) or to both sides therefore there are no non-trivial $C$-local conservation laws hosted. This behaviour cannot happen by embedding an Abelian algebra. Notably, walls can have more general free-fermion dynamics in the centre by having an arbitrary matchgate circuit (which is the automorphism group of free fermion operators) on $C$ with conditional couplings to $L, R$ based on projectors to irreducible subspaces \cite{Valiant2002SIAMMatchgates, TerhalDiVincenzo2002NoninteractingFermions}. The $\mathrm{FSWAP}$ wall is the simplest example of this behaviour.

\section{Locally constrained many-body dynamics  \label{sec:dynamical_features}}

\subsection{Operator space fragmentation}

In this section, we link the observed splitting of operator space to Hilbert-space fragmentation studied in many-body physics \cite{Moudgalya2022} in the context of random circuit dynamics. 
Due to bounded light-cones, operator space splits into spatially localised invariant sectors $\mathcal{L}, \mathcal{R}$ which induces a splitting of the central operator space to the intersection $\mathcal{I}$:
\begin{align}
    \mathcal{I} &= \overline{\mathcal{M}_L} \cap \overline{\mathcal{M}_R} =\mathcal{A}_C \cap \mathcal{B}_C, 
    \\
    %\mathcal{C} &= \mathcal{B}(C) \setminus \mathcal{I} \\
    \mathcal{L} &= \overline{\mathcal{M}_L}\setminus \mathcal{I} \supseteq \mathcal{M}_L, \\
    \mathcal{R} &= \overline{\mathcal{M}_R} \setminus \mathcal{I} \supseteq \mathcal{M}_R.
\end{align}
\noindent This implies the following decomposition of operator space:
\begin{equation}
    \mathcal{M}_{LCR} = \mathcal{F} \oplus \mathcal{F}^{\perp} = \mathcal{L} \oplus \mathcal{R} \oplus \left (\mathcal{L} \times \mathcal{R}\right) \oplus \mathcal{I} \oplus \mathcal{F}^{\perp}.
\end{equation}
\noindent  Fragmented circuits are defined to have of an exponentially large sector of invariant subspace in the dynamics. The wall induces splitting of the many-body operator space so that individual left and right localised subsystems can be multiplied to create larger invariant spaces. In random circuits, this leads to fragmentation whenever the circuit hosts an extensive number of bounded light cones.

A specific instance of fragmented circuit ensemble was investigated in detail in our previous work \cite{Kovacs2026} where the unitaries are chosen uniformly randomly of entangling Clifford gates with a finite probability of non-Clifford perturbations. In the large system limit, randomly gates will act invariantly on, for example, diagonal algebras that lead to fragmentation. We stipulate that any gateset which has a finite probability of satisfying the wall constraint (e.g. containing conditional unitaries as defined in \Cref{app:conditionals}) will lead to phenomenologically equivalent non-ergodic evolution, with an extensive set of walls and consequently exponentially sized fragment space. We note, however, that due to the dimensional reduction required for sub-algebra embedding to satisfy the wall constraint, we expect that fragmented circuit ensembles always form a zero-measure set of all unitaries, independent of circuit connectivity.  

\subsection{Entanglement area law \label{sec:entanglement}}

A wall consisting of entangling gates requires the embedding of a sub-maximal dimensional algebra in the circuit. This is manifest in the restricted entanglement growth across the central region due to sub-maximal operator (and state) Schmidt rank. 

\begin{theorem}[Bounded entanglement]
Consider the product state $\ket\alpha_L \ket\beta_{CR}$ evolving under a left-wall unitary $U$ with invariant algebra $\mathcal{M}_L \otimes \mathcal{A}_C \otimes \Id_R$. The Schmidt rank of the evolved state with respect to the partition $L-CR$ is bounded by
\begin{align}\label{sr Fab}
  \mathrm{Sr}(U^t \ket\alpha_L \ket\beta_{CR})  
  \leq \mathrm{dim}(\mathcal A_C) \text{ for all } t.
\end{align}
\label{thm:schmidt_rank_bound}
\end{theorem}

\begin{proof}
Denote the irreducible subspaces of the invariant algebra $\mathcal{A}_C$ as $\mathcal{D}_i, \mathcal{E}_i$ under \Cref{thm:wall_structure}.
Let $\ket{\beta_i}_{\mathcal D_i \mathcal E_i R}$ be the projection of $\ket\beta_{CR}$ onto the subspace $\mathcal D_i \otimes \mathcal E_i \otimes R$, and let 
\begin{align}
  \ket{\beta_i}_{\mathcal D_i \mathcal E_i R} = \sum_j \lambda_j \ket{\mu_{i,j}}_{\mathcal D_i} \ket{\nu_{i,j}}_{\mathcal E_i R}        
\end{align}
be its Schmidt decomposition, whose rank is at most $\mathrm{dim}\mathcal D_i$. 
After the action of $U^t$, the state $\ket\alpha_L \ket{\mu_{i,j}}_{\mathcal D_i} \ket{\nu_{i,j}}_{\mathcal E_i R}$ evolves into $\ket{\tilde \mu_{i,j}}_{L\mathcal D_i} \ket{\tilde\nu_{i,j}}_{\mathcal E_i R}$, whose Schmidt rank with respect to the partition $L-\mathcal D_i \mathcal E_i R$ is the same as the Schmidt rank of $\ket{\tilde \mu_{i,j}}_{L\mathcal D_i}$, which is at most $\mathrm{dim}\mathcal D_i$. This implies that the Schmidt rank of 
\begin{align}
  U^t \ket\alpha_L \ket{\beta_i}_{\mathcal D_i \mathcal E_i R}
  = \sum_j \lambda_j \ket{\tilde \mu_{i,j}}_{L \mathcal D_i} \ket{\tilde \nu_{i,j}}_{\mathcal E_i R}
\end{align}
is at most $\mathrm{dim}^2\mathcal D_i$. Which implies that the Schmidt rank of 
\begin{align}
  U^t \ket\alpha_L \ket{\beta}_{CR}
  = \sum_{i,j} \ket{\tilde \mu_{i,j}}_{L \mathcal D_i} \ket{\tilde \nu_{i,j}}_{\mathcal E_i R}
\end{align}
is bounded by $\sum_i \dim^2 \mathcal{D}_i = \dim \mathcal{A}_C$.
\end{proof}

There are several remarks in order. We have showed the entanglement bound under a left wall evolution while an entirely similar calculation establishes a bound for the right wall evolution using \Cref{thm:left_right_closure}. For a maximal Abelian sub-algebra, the above bound saturates to $\dim C$ which restricts the attainable entanglement across the $L-CR$ and $LC-R$ by the central \enquote*{bottleneck} subsystem dimension. For wall unitaries which have random blocks (drawn under the Haar-measure for each invariant subspace) we expect the bound to be saturated up to the Page correction of random state entanglement \cite{Page1993AverageEntropy}. 

For non-abelian embedded algebras, we expect the bound to be looser, and the Schmidt rank bounded by the dimension of the maximal Abelian sub-algebra of $\mathcal{A}_C$. The operator localisation implies that $L, R$ subsystems share no mutual information while each can be entangled with a third-party highlighting the subtle differences between operator spreading and entanglement spreading in unitary dynamics. Operator dynamics under a wall unitary is therefore an instance of an area-law phase \cite{Eisert2010}, since scaling the dimensions of $L, R$ cannot increase the state entanglement.

\subsection{Measurement-induced dynamics}
\begin{figure}
    \centering  \includegraphics[width=0.5\linewidth]{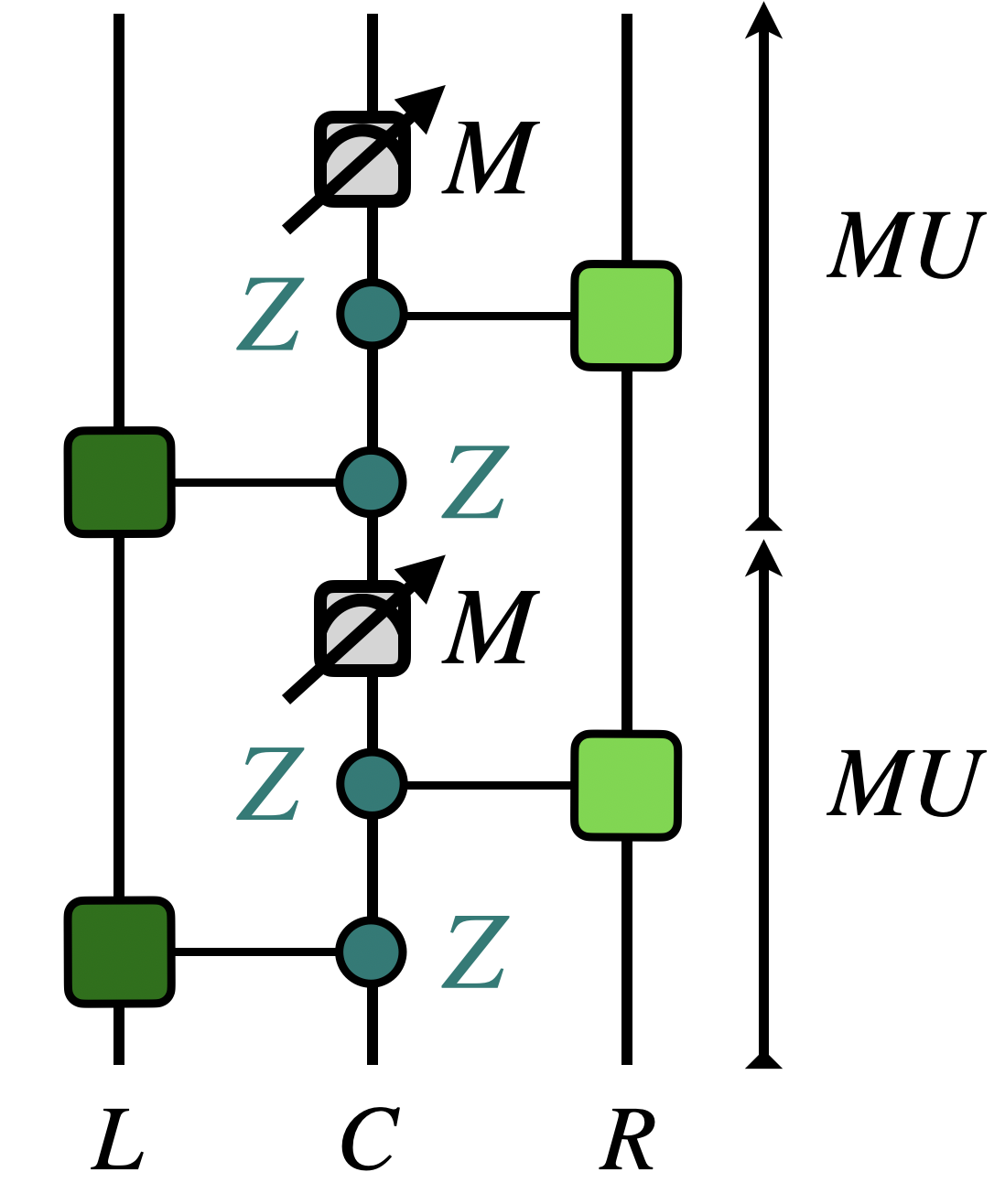}
    \caption{Mesurement protocol with unitary evolution followed by a $C$-local projective measurement. The wall dynamics leads to area-law entanglement. The invariant structure of the wall determines whether the local measurement preserves, decreases or increases entanglement.}
    \label{fig:wall_plus_measure}
\end{figure}
 
In this section, we consider the effect of central projective measurement on the entanglement generation under wall dynamics. As a setup, we consider evolving a $L-CR$ product state under the wall dynamics (as in \Cref{thm:schmidt_rank_bound}) and measuring afterwards. Due to the  splitting of operator space, the effect of measurement is determined from the compatibility of the measured operator with the invariant algebra of the wall. 

If $M_C \in \mathcal{A}_C$, the state is projected to one of the invariant blocks of the embedded algebra, therefore the Schmidt rank of the state satisfies: $\mathrm{Sr(\Ad_{M} \Ad_U(\ket\alpha_L\ket\beta_{CR}}) \leq \dim^2 \mathcal{D}_i$ where $\mathcal{D}_i$ is determined from the measurement outcome. If the wall unitary is Haar-random within irreducible blocks, this is in accordance with the conventional notion of measurement being detrimental to entanglement generation in the circuit \cite{Fisher2023, Li2018, SkinnerRuhmanNahum2019MeasurementTransition}. If $M_C \in \comm(\mathcal{A}_C)$, the state is undisturbed by the measurement and its entanglement is unchanged. 

Now consider $M_C \in \mathcal{M}_C \backslash(\mathcal{A}_C \cup \comm(\mathcal{A}_C))$. These measurements break the splitting of operator space and force the evolved state into the commutant space of $M$. Analogous to entanglement swapping protocols \cite{ZukowskiZeilingerHorneEkert1993EntanglementSwapping, PanBouwmeesterWeinfurterZeilinger1998ExperimentalEntanglementSwapping}, the measurement entangles the initially commuting subspaces of $C$, leading to the entangling of $L, R$ subsystems without the bottleneck system. In this case, the state can exhibit \textit{volume-law} entanglement as a result of the local measurement by iterating unitary evolution interleaved with a local measurement (see \Cref{fig:wall_plus_measure}). 

Concrete examples of random circuit dynamics interleaved with measurement in the presence of bounded constraints thus promises rich measurement-induced dynamics that we leave for future work. Recent work investigated this behaviour using controlled-$Z$ gates with random $X$ or $Z$ measurements to model the stability of localised circuit evolution against local measurements \cite{IppolitiGullansGopalakrishnanHuseKhemani2021MeasurementOnly}. The analytical theory developed in this paper provides the roadmap to generalisations of these results with the algebraic characterisation determining which measurements are stable and which are not. Finally, understanding the effect of generalised measurements in the presence of operator space fragmentation is an intriguing direction for future analysis.
\subsection{Spectral correlations}

This section considers the spectral form factor (SFF) in the presence of wall constraints for a tripartite random wall ensembles. This quantity is a well-established probe of quantum chaos and the exact operator space decomposition facilitates its analytical calculation for bounded-light cone dynamics \cite{Brezin1997, Haake1999, Haake2010, Chan2018MBQC, Chan2018, Bertini2021}. Consider drawing gates from a random unitary ensemble $U \sim{\mathbf{U}}$. The SFF is defined as the ensemble-average trace powers:
\begin{equation}
K(t) = \mathds{E}_{U \sim \mathbf{U
}} \left [|\tr(U^t)| ^2\right]. 
\end{equation}
\noindent From random-matrix theory, it is well-known that taking $\mathbf{U}$ as $d$-dimensional Haar-random unitaries leads to the SFF $K_{\mathrm{Haar}}(t) = \min(t, d)$ \cite{Mehta2004, Haake2010}. 

Consider a random ensemble of walls with randomised invariant blocks in \Cref{thm:wall_structure}, $T^i \sim \mathrm{Haar}(U(L\mathcal{D}_i)), R^i \sim \mathrm{Haar}(U(L\mathcal{E}_i))$. Since we are interested in trace powers, any permutation of the invariant blocks leave SFF invariant.  It is easily shown the form factor is multiplicative under ensembles of the form $U_1 \otimes U_2$ where $U_1, U_2$ are i.i.d. drawn random unitaries and additive over irreducible blocks.
This implies that for the wall ensemble, 
\begin{equation}
    K(t)  = \sum_{i = 1}^{n_i} \min(t, \dim L\mathcal{D}_i) \min(t, \dim \mathcal{E}_iR), 
\end{equation}

\noindent where we note the dimensional constraint $\sum_{i=1}^{n_i} \dim\mathcal{D}_i \dim\mathcal{E}_i = \dim C$. In the presence of a single bounded light cone, early time spectral correlations scale as $K(t)\sim t^2$ indicating the effect of symmetry. As an example, consider Abelian wall with a maximal diagonal algebra (\Cref{sec:abelian_walls}) with $\dim L = \dim R = d$ and $\dim C = d_C$. The spectral form factor is $K(t) = d_C \min(t, d)^2$ indicating that the saturation (Heisenberg) time scales as $\tau \sim \sqrt{d}$ separating the fragmented ensemble for a uniform random ensemble with $K_{\mathrm{Haar}}(t) = \min(t, d_C d^2)$. For a non-abelian embedded algebra with full matrix invariant blocks, we expect that $K(t) \sim t^2$ for early times and crosses over to $K(t) \sim t$ at a timescale set out by the size of the largest invariant block dimension. 

Consider now an ensemble with extensively many walls, where the operator space is fragmented. In this case, we expect the spectral form factor to be super-polynomial $K(t) \sim \lim_{n\rightarrow\infty} t^n$ in the thermodynamic limit which highlights the difference between an integrable form factor $K(t) = d$ and a putative many-body localised one $K(t) \sim \exp(t)$ \cite{Chan2018, Chan2021SpectralLyapunov, Chan2018}. We leave for future numerical investigations whether operator space fragmented circuits in this sense are sufficient to create many-body integrable spectral correlations. Finally, we note that perturbing a bounded light cone by a non-wall unitary can lead to slow spectral relaxation e.g. in Clifford circuits with Haar-random perturbations \cite{Kovacs2026}. By breaking up invariant subspaces, we expect that at late times, $K(t) \approx t$ and that the signature localisation (fast growth of SFF) is only visible for early times. It would be intriguing to compare our work with cross-over behaviour in systems with approximately broken symmetries exhibiting \enquote*{Hilbert-space diffusion} \cite{BaumgartnerDelacretazNayakSonner2024HilbertSpaceDiffusion} and to calculate analytical bounds on the crossover timescale.   

\section{Conclusion}

This paper investigated the general algebraic structure of tripartite unitary evolution that bounds the causal light cone in time-periodic operator dynamics. We understood this phenomenon through the splitting of operator space into commuting sectors from the invariance of an embedded sub-algebra on a subsystem. The wall unitaries which exhibit causal independence were then constructed as the unitary representations of the invariant algebra's automorphism group.  The algebraic framework naturally lended itself to characterise the local conservation laws of these unitaries based on commutant structure of the invariant algebra. Our construction is an example of analytically tractable ergodicity breaking dynamics which is robust to sparse local perturbations.

This work generalises the author's previous results on this topic in time-periodic random Clifford circuits by constructing the most general unitaries exhibiting the associated phenomenology of operator space fragmentation \cite{Kovacs2026}.  We have constructed wall unitaries without non-trivial central conservation laws using non-commutative central algebras. These highlight that locally constrained non-ergodic quantum dynamics is possible without solitons or integrability since the central subsystem may evolve ergodically within the commuting operator subspaces. Our model is thus a simple etalon for Hilbert-subspace ergodicity with the invariant subspaces being spatially localised \cite{Logaric2025}.

Some observable probes in many-body ergodicity were then characterised. We have proven an entanglement area-law due to the bounded light-cone due to the reduction of maximal dimension of the invariant algebra without relying on the stabiliser formalism used in the Clifford-case \cite{Kovacs2026}. The formalism also lends itself to understand the stability of localisation against perturbations and how bounded light cones affect chaotic spectral correlations.  Our results are a rigorous study of locally constrained dynamics and we hope that the formalism here will enable the analytic study of circuit toy models to observe novel measurement-induced phase transitions in the presence of non-ergodic evolution. 

There are several theoretical directions for generalising the theory presented. All our results focused on finite-dimensional algebras which immediately raises the question how much of the phenomenology persists in infinite dimensions. For example, we can conceive that the equivalence of restricted operator spreading from left-to-right and right-to-left in \Cref{thm:left_right_equiv} may not hold in infinite systems which could lead to operator dynamics with \textit{one-sided} causality: ie. left subsystem being causally independent of right subsystem but not vice versa. In this work, we focused on tripartite systems which is the simplest setting in which non-trivial causal decoupling can occur while it remains as an open problem how such causal decompositions can be constructed in multi-partite systems.

We have pre-dominantly focused on discrete and time-periodic operator evolution with simple generalisations to time-dependent dynamics through gauge transformations of an underlying wall unitary evolution. Understanding the conditions under which continuous time (e.g. Hamiltonian) evolution exhibits causal independence remains an interesting open problem. Deriving the structure of a Hamiltonian's Lie algebra required to fix invariant matrix algebras could enable the systematic construction of many-body Hamiltonians exhibiting causal independence.

Finally, we conjecture that given a unitary circuit, deciding whether it exhibits a bounded light cone can be efficiently done on a quantum computer for finite times. Encoding a left-local operator as a Choi-state, the arrested operator evolution would be manifest as the separability of the time-evolved state. Efficient schemes exist for testing separability of multi-partite systems \cite{Bouland2024} which could be utilised for this purpose. Given that a bounded light-cone leads to an area-law, we hypothesise that efficient tensor network representations also exist for the unitaries studied. Our model may be utilised as a simple benchmarking scheme for complex dynamics in quantum simulators which could lift the impact of our analysis to practical applications.

\section*{Acknowledgements}

M.D.K.~gratefully acknowledges Augustin Vanrietvelde, Maxwell West, Ruben Ibarrondo-Lopez, Ricard Puig, Aniruddha Sen and Surajit Bera for useful discussions on related topics. M.D.K.~wishes to thank Marco Schiró and his research group for their hospitality at College de France where part of this work was undertaken. M.D.K.~was supported by the UK Engineering and Physical Sciences Research Council (EPSRC) [grant no.  EP/S021582/1], by the Los Alamos Quantum Computing Summer School Fellowship Program and by the Scientific High-Level Visiting Fellowship of the UK-French Embassy. C.J.T.~was supported by an EPSRC fellowship (Grant Ref. EP/W005743/1). \textit{Statement of AI use.} In the preparation of this manuscript, ChatGPT 5 and Claude AI (Fable 5.0 and Opus 5.0) has been used for drafting and literature review and aid for prooving \Cref{lemma:subalgebra_product}.The authors take full responsibility for the results presented. \textit{Data avalability statement.} This work is theoretical in nature and does not require supporting data.
\appendix
\section*{Appendix}

\section{$\mathrm{C}^*$-algebras \label{app:algebras}}
We review the elementary properties finite-dimensional $\mathrm{C}^*$-algebras and provide additional mathematical background for the results in this work. For a more complete, yet accessible, review, the Reader is referred to \cite{Harlow2017, Landsman1998_notes, Jones2020, Bény2020}.

\begin{definition}[Algebra over $\mathds{C}$]
 An algebra $A$ over the field of complex numbers $\mathds{C}$ is defined as a vector space $V[\mathds{C}]$ which is closed under a bilinear product operation, that is, for $a, b \in A$ and $\alpha, \beta \in \mathds{C}$:
 \begin{align}
     &\alpha a + \beta b \in A, \\
     & a b \in A.
 \end{align}
 \end{definition}

\begin{definition}[Generators]
    The notation $\mathcal{A} =\mathrm{Alg}\{g_1, ..., g_n \}= \langle g_1, ..., g_n \rangle$ is the algebra generated by arbitary products of $g_i$ and arbitrary complex linear combinations of their products.
\end{definition}

\begin{definition}[$\mathrm{C}^*$-algebra]
    We define a $\mathrm{C}^*$-algebra as an algebra defined over a Banach space with an additional involution operation $*$, which satisfies:    \begin{equation}
        ||x^* x|| = ||x||^2 \text{ for all } x \in \mathcal{A}.
    \end{equation}
\end{definition}

\begin{definition}[Commutative algebras]
    An algebra $\mathcal{A}$ is commutative or Abelian, if $[a, b] = 0$ for every $a, b \in \mathcal{A}$.
\end{definition}

\begin{definition}[Full matrix algebra]
    Let $\mathcal{M}_S$ denote the $d \times d$ complex matrices acting on a Hilbert space $\mathcal{H}_S$ with $\dim \mathcal{H}_S = d$.
\end{definition}
\begin{proposition}
    $\mathcal{M}_S$ forms a \cstar under matrix multiplication as the product operation, the adjoint as the involution and, for example, the Hilbert-Schmidt norm $||A|| = \tr(A^{\dagger}A)$ as the norm. Note any other norm would suffice as in finite-dimensional vector spaces, all norms are equivalent.
\end{proposition}
\begin{definition}[Bounded linear operators]
    Let $\mathcal{B}(S)$ denote the set of bounded linear operators on the Hilbert space $\mathcal{H}_S$ of a quantum system $S$. These are defined as all linear maps $f:\mathcal{H}_S \rightarrow \mathcal{H}_S$ so that there exists finite $K\in \mathds{R}$:
    \begin{equation}
        ||f(\xi)||\le K||\xi|| \text{ for every } \xi \in \mathcal{H}_S.
    \end{equation}
\end{definition}
\begin{definition}[von-Neumann algebras]
    A unital sub-algebra of $\mathcal{B}(S)$ (containing an identity element) is called a von-Neumann algebra.
\end{definition}

In this work, we are only concerned with finite-dimensional unital $\mathrm{C}^*$-algebras and therefore we use the term algebra, von-Neumann algebra and \cstar interchangeably. Differences arise between these in infinite-dimensional systems which go beyond the scope of this work.

\section{Commutants and centres \label{app:commutants}}
\begin{definition}[Commutant]
    The commutant of an algebra $\mathcal{A}$ is defined as the set of operators commuting with elements of $\mathcal{A}\subseteq \mathcal{M}_S $ in a system $S$.
    \begin{equation}
        \mathrm{Comm}(\mathcal{A})_S = \{ c \in \mathcal{M}_S |[c, a] = 0 \text{ }, \forall a \in \mathcal{A}\}.
    \end{equation}
    We also denote $\mathcal{A}'$ as the commutant. To ease notation, we drop the subscript $S$ unless the commutant is defined over a proper sub-algebra of $\mathcal{M}_S$.
\end{definition}
\begin{proposition}[Commutants are algebras]
    The commutant of sub-algebra $\mathcal{A} \subseteq \mathcal{M}_S$  is also a $\mathrm{C}^*$-algebra.
\end{proposition}

\begin{proof}
    We show closure of the commutant elements under the algebraic operations. 
    
    If $c_1, c_2 \in \comm(\mathcal{A}), a \in \mathcal{A} , \alpha, \beta \in \mathds{C}$, 
    \begin{align}
        &[\alpha c_1 + \beta c_2, a] = \alpha [c_1, a] + \beta [c_2, a] = 0 \text{ } \forall a \in \mathcal{A}, \\
        & [c_1c_2, a] = c_1[c_2, a] + [c_1, a]c_2 = 0\text{ } \forall a \in \mathcal{A}.
    \end{align}
    Since commutators are elements of the full matrix algebra $\mathcal{M}_S$, the other axioms of $C^*$-algebras automatically apply.
\end{proof}
\begin{definition}[Centre]
    A centre of an algebra $\mathcal{A}\subseteq \mathcal{M}_S$ is defined as:
    \begin{equation}
        Z(\mathcal{A}) = \mathcal{A} \cap \comm(\mathcal{A}).
    \end{equation}
\end{definition}
\begin{corollary}
    $Z(\mathcal{A})$ is also a \cstar.
\end{corollary}

\begin{definition}
    An algebra $\mathcal{F}\subseteq\mathcal{M}_S$ is called a \textit{factor} if its centre is trivial: $Z(\mathcal{F}) = \langle \Id \rangle$.
\end{definition}
\begin{proposition}
    Let $\mathcal{F}\subseteq \mathcal{M}_S$ be a factor. Then, there exists a unitary $V\in \mathcal{M}_S$ so that:
    \begin{equation}
        \mathcal{F}= \Ad_{V^{-1}}(\Id_{\mathcal{D}} \otimes \mathcal{M}_{\mathcal{E}}), 
    \end{equation}
    where $\mathcal{D, E}$ are subspaces satisfying $\dim \mathcal{D} \dim \mathcal{E} = \dim \mathcal{H}_S$ and $\dim \mathcal{M}_{\mathcal{E}}=(\dim \mathcal{E})^2 = \dim \mathcal{F}$.
\end{proposition}
\begin{proof}
    See Lemma 2.2 in \cite{Bény2020}.
\end{proof}
\begin{proposition}
    A commutative algebra is contained in its commutant.
\end{proposition}
\begin{proof}
    Follows from the definition. Equivalently, commutative algebras are their own centres (self-centred).
\end{proof}

\begin{theorem}[Commutant of a tensor product]
    Let $\mathcal{A}_{P}, \mathcal{B}_Q \subseteq \mathcal{M}_{PQ}$ for a bi-partite system $PQ$. Then $\comm(\mathcal{A} \otimes \mathcal{B)}$ = $\comm(\mathcal{A}) \otimes \comm(\mathcal{B}).$
\end{theorem}
\begin{proof}
    See \cite{Rousseau1976} and references therein.
\end{proof}

\begin{proposition}
    $\comm(\mathcal{M_S}) = \langle \Id \rangle$.
\end{proposition}
%\begin{proof}
%    \marcell{Schur's lemma?}
%\end{proof}
\begin{proposition}
    $\comm(\langle \Id \rangle) = \mathcal{M}_S$.
\end{proposition}
\begin{theorem}[Double commutants]
    If $\mathcal{A}$ is a von-Neumann algebra, then $\comm(\comm(\mathcal{A})) = \mathcal{A}$.
\end{theorem}
\begin{proof}
    See Theorem A.3. in appendix of \cite{Harlow2017} for the finite-dimensional case and \cite{Landsman1998_notes} in general.
\end{proof}
\section{Structure of sub-algebras \label{app:sub_algebra_structure}}
\begin{lemma}[Structure of sub-algebras]
    Let $L, C$ be subsystems of  $\mathcal{H}_{LC}$. Let $\mathcal R$ be an algebra satisfying
    \begin{equation}
        \mathcal{M}_L \subseteq \mathcal{R} \subseteq \mathcal{M}_{LC},
    \end{equation}
    then there exists $\mathcal{A}_C \subseteq \mathcal{M}_C$ such that 
    \begin{equation}
        \mathcal{R} = \mathcal{M}_L \otimes \mathcal{A}_C\text{.}
    \end{equation}
    \label{lemma:subalgebra_product}
\end{lemma}
%\marcell{ChatGPT proof -- take caution!}
\begin{proof}
    Since $\mathcal{M}_L \subseteq \mathcal{R}$, we have $\mathrm{Comm}(\mathcal{R}) \subseteq \mathrm{Comm}(\mathcal{M}_L \otimes \Id_C) = \Id_L \otimes \mathcal{M}_C$.
    Hence, there exists some algebra $\mathcal{C}_C$ such that $\mathrm{Comm}(\mathcal{R}) = \Id_L \otimes \mathcal{C}_C$.
    From the double-commutant theorem in finite-dimensions, 
    it follows that $\mathcal{R} = \mathrm{Comm}( \Id_L \otimes \mathcal{C}_C) = \mathcal{M}_L \otimes \mathcal{A}_C$, where $\mathcal{A}_C$ is the commutant of $\mathcal{C}_C$ relative to the $C$ subsystem.
\end{proof}
\begin{theorem}[Structure of commutative algebras]
    An Abelian \cstar $\mathcal{A}$ is spanned by $\dim \mathcal{A}$ mutually orthogonal projectors:
    \begin{equation}
        \mathcal{A} = \langle \Pi_1, ..., \Pi_{\dim\mathcal{A}}\rangle,
    \end{equation}
    \label{thm:abelian_structure}
    \noindent where $\Pi_i\Pi_j = \delta_{ij}\Pi_j$.
\end{theorem}
\begin{proof}
    See Theorem 2.6 of \cite{Bény2020}.
\end{proof}

\begin{theorem}[Structure of finite algebras]
    Let $\mathcal{A}_S \subseteq \mathcal{M}_S$ be a \cstar. Then, $\mathcal{A}_S$ decomposes as a direct sum of factors under a unitary equivalence: 
    \begin{equation}
        \mathcal{A}_S \cong \bigoplus_i\mathcal{F}_i, 
    \end{equation}
    \noindent where each $\mathcal{F}_i$ is a factor. The dimensions satisfy $\sum_i\dim \mathcal{F}_i = \dim \mathcal{A}_S$.
\end{theorem}
\begin{proof}
    See \cite{Bény2020, Pedersen2018AutomorphismGroups}.
\end{proof}

\section{Representation theory of finite-dimensional algebras \label{app:rep_theory}}
\begin{lemma}[Block-structure of a $\mathrm{C}*$-algebra]
    Let $\mathcal{A}_S\subseteq \mathcal{M}_S$ be a \cstar over Hilbert space $\mathcal{H}_S$. Then, under the action of $\mathcal{A}_S$, the Hilbert space decomposes as a block sum: 
    \begin{equation}
        \mathcal{H}_S\cong \bigoplus_i \mathcal{H}_{\mathcal{D}_i} \otimes \mathcal{H}_{\mathcal{E}_i},
    \end{equation}
    \noindent where $\mathcal{D}_i, \mathcal{E}_i$ are subspaces hosting irreducible matrix blocks. The representation of the algebra decomposes as:
    \begin{equation}
        \mathcal{\mathcal{A}_S} \cong \bigoplus_i \mathcal{M}_{\mathcal{D}_i} \otimes \mathds{1}_{\mathcal{E}_i},
    \end{equation}
    \noindent where $i$ labels the irreducible representations, $D_i = \dim \mathcal{D}_i$ is the dimension of the subspace and $ \dim \mathcal{E}_i = E_i$ the dimension of degeneracy space. We have the constraints:
    \begin{align}
         \dim \mathcal{H}_S&= \sum_i D_i E_i, \\
         \dim \mathcal{A_S}&= \sum_i\dim  \mathcal{M}_{\mathcal{D}_i} = \sum_i ( D_i)^2 .
    \end{align}
    
\end{lemma}
\begin{proof}
    See Theorem 2.7 of \cite{Bény2020}.
\end{proof}

\section{Structure of the unitary automorphism group \label{app:normaliser}}

\begin{definition}[Unitary automorphisms]
    Let $\mathcal{A}_S \subseteq \mathcal{M}_S$ be a \cstar. 
    We define $\mathrm{U}(\mathcal{A}_S)$ as the set of unitaries in $\mathcal{M}_S$ that leave the algebra invariant:
    \begin{align*}
        \mathrm{U}(\mathcal{A}_S) = &\{U\in \mathcal{M}_S | UU^{\dagger} = U^{\dagger} U = \Id \\ &\text{ and } \Ad_U(a) \in \mathcal{A}_S \text{, } \forall a \in \mathcal{A}_S \}.
    \end{align*}

    \noindent We also refer to $\mathrm{U}(\mathcal{A}_S)$ as the normaliser of $\mathcal{A}_S$. An element of $\mathrm{U}(\mathcal{A}_S)$ is called internal if it lies in $\mathcal{A}_S$, otherwise it is external.
\end{definition}
\begin{lemma}
    The unitary automorphisms form a group.
\end{lemma}
\begin{proof}
    The normaliser $\mathrm{U}(\mathcal{A})$ is the intersection of all automorphisms of $\mathcal{A}$ and the unitary group.
    The set of all automorphisms naturally forms a group.
    The intersection of two groups is itself also a group.
\end{proof}
\begin{lemma}[Normaliser of commutant]
    \begin{equation}
      \mathrm{U}(\mathcal{A}_S) = \mathrm{U}(\comm(\mathcal{A}_S)).
    \end{equation}
\end{lemma}
\begin{proof}
    % We will show that $U \in \mathrm{U}(\mathcal{A}_S)$ if and only if $U \in \mathrm{U}(\comm(\mathcal{A}_S))$.
    Let $a \in \mathcal{A}_S$, $b \in \comm (\mathcal{A}_S)$ and $U \in \mathrm{U}(\mathcal{A}_S)$.
    From the definition of the normaliser, we have $[\Ad_{U^{-1}}(a),b] = 0$ and hence, $[a, \Ad_{U}(b)] = 0$.
    Since $a$ was arbitrary $\Ad_{U}(b) \in \comm(\mathcal{A}_S)$ and therefore $U \in \mathrm{U}(\comm(\mathcal{A}_S))$.
    % For the sufficient condition, let $U\in \mathrm{U}(\mathcal{A}_S)$. Then, we have $[\Ad_U(a), b] = 0$ from the property of normaliser. This implies $[a, \Ad_{U^{-1}}(b)] = 0$. For every element $U$ of the normaliser, $U^{-1}$ is also an element, therefore the images of the commutant elements (of the form $\Ad_U(b)$) lie in the commutant of $\mathcal{A}_S$. Therefore $U$ is in the normaliser of $\comm(\mathcal{A}_S)$.
    For the reverse inclusion, use $\comm(\mathcal{A}_S)$ in place of $\mathcal{A}_S$ in the previous argument and use the double commutant theorem $\comm(\comm(\mathcal{A}_S)) = \mathcal{A}_S$ as appropriate.
    % For the necessary direction, we start with $U \in \mathrm{U}(\comm(\mathcal{A}_S)$. Then, $[a, \Ad_U(b)] = 0$ which implies that  $[\Ad_{U^{-1}}(b), a] = 0$. For every element $U^{-1}$, there is en element of the normaliser $U$. Therefore, the images of $\mathcal{A}_S$ elements under the automorphism (of the form $\Ad_U(a)$) lie inside $\comm(\comm(\mathcal{A}_S)) = \mathcal{A}_S$ under the double commutant theorem. Therefore $U$ is in the normaliser of $\mathcal{A}_S$.
\end{proof}
\begin{corollary}
    \begin{equation}\mathrm{U}(\mathcal{A}_S) = \mathrm{U}(\mathcal{A}_S) \cap  \mathrm{U}(\comm(\mathcal{A}_S)).
    \end{equation}
\end{corollary}
\begin{lemma}[Normaliser of the centre]
    \begin{equation}\mathrm{U}(\mathcal{A}_S) \subseteq \mathrm{U}(Z(\mathcal{A}_S)).\end{equation}
    \label{lemma:comm_normaliser}
\end{lemma}
\begin{proof}
    For every $z \in Z(\mathcal{A}_S)$ and $a \in \mathcal{A}_S$, it holds that $[a,z] =0$. Then, $[\Ad_U(a), z] = 0$ using the property of the normaliser. It follows that $[a, \Ad_U^{-1}(z)] = 0$. Since the unitaries of $\mathrm{U}(\mathcal{A})$ form a group we have $[a, \Ad_U(z)] = 0$ for every $U \in \mathrm{U}(\mathcal{A}_S)$. Therefore the image of $z$ under the automorphism commutes with every element of $a$, therefore it must be an element of the centre.
\end{proof}

Now we would like to better understand the index set for the direct sum in the decomposition.
% Minimal projectors -> joint representations
A minimal central projection $z \in \min(\mathcal{A})$ is a non-null central idempotent element which contains no proper subprojection.
These are connected to the decomposition of $\mathcal{A}$ as follows,
\begin{align}
    \mathcal{A} = \sum_{z \in \min \mathcal{A}} z \mathcal{A} \cong \bigoplus_{z \in \min \mathcal{A}} (\mathcal{M}_{\mathcal{D}_z} \otimes \mathds{1}_{\mathcal{E}_z})
    \label{eq:decomp_min}
\end{align}
where the equality is a canonical isomorphism and the $\cong$ denotes spatial isomorphism, rather than abstract isomorphism.
The minimal central projections of $\mathcal{A}$ are the same as those for $\mathcal{A}'$; a fact which underlies how they may be simultaneously decomposed.
To each $z$ we can identify a subrepresentation $\pi_z$, using the first isomorphism theorem for modules for the module homomorphism which multiplies by $z$.
This is both a representation of $\mathcal{A}$ and $\mathcal{A'}$, a property to which we refer as being a joint representation.
It is also an \emph{irreducible} joint representation, meaning that it contains no proper non-trivial subspace which is also a joint representation.
Equivalently, it is an irreducible representation of the product algebra $\mathcal{A}\mathcal{A'}$.
For $x,y \in \min(\mathcal{A})$, the joint representations $\pi_x$ and $\pi_y$ are isomorphic (as joint representations) if they are both isomorphic as $\mathcal{A}$-representations and as $\mathcal{A}'$-representations.
This is however not equivalent to being isomorphic as $\mathcal{A}\mathcal{A}'$-representations.

% Equivalence classes and permutation groups of joint representations
We can form a permutation group $\Sigma$ over $\min \mathcal{A}$ for which the joint-isomorphism classes are the conjugacy classes.
After fixing a spatial isomorphism in \cref{eq:decomp_min}, isomorphic joint-representations are essentially equal (canonically isomorphic).
In the finite-dimensional case, each joint-representation can then be characterised by two numbers $D_z$ and $E_z$, the dimensions of the subspaces $\mathcal{D}_z$ and $\mathcal{E}_z$.
Hence, a pair of joint-representations are isomorphic if and only if they agree on these two dimensions.
This provides $\Sigma$ with a representation on the underlying state space $\mathcal{H}$ as an operator that rearranges the terms of the direct-sum decomposition.
In this way, $\Sigma$ can be identified with a subgroup of $\mathrm{U}(\mathcal{A})$ and this allows us to define a coset space $\mathrm{U}(\mathcal{A}) / \Sigma$, consisting of the equivalence classes of $\mathrm{U}(\mathcal{A})$ under right-multiplication by $\Sigma$.

\begin{lemma}[Representative elements of the coset space]  
    Within every coset in $\mathrm{U}(\mathcal{A})/\Sigma$ there is a unique representative element which leaves $Z(\mathcal{A})$ invariant, together these form a set of representatives $V$. This means that any automorphism $U \in \mathrm{U}(\mathcal{A})$ may be uniquely decomposed into a representative element $v \in V$ and a permutation $\sigma \in \Sigma$ such that $U = v \sigma$.
    \label{lemma:U_unique_decomposition}
\end{lemma}
\begin{proof}
  Let $U \in \mathrm{U}(\mathcal{A})$.
  Since $\mathrm{U}(\mathcal{A})$ is the unitary automorphism group of $\mathcal{A}$, the image of $Z(\mathcal{A})$ under $U$ is $Z(\mathcal{A})$ itself.
  Recall that the centre $Z(\mathcal{A})$ has a basis of projection operators, $\Pi_i$ for $i \in I$, corresponding to each term of the direct sum decomposition of $\mathcal{H}$.
  Examine the action of $U$ on this basis,
  \begin{align}
      \Ad_U : \Pi_i \mapsto \sum_{j \in J} \Pi_j & \text{ for some $J \subseteq I$.}
  \end{align}
  since it must preserve the eigenvalues of $\Pi_i$.
  If we take some element of the basis $\Pi_i$ such that the trace of $\Pi_i$ is minimal, then $J$ consist of a single element $j$.
  Since every other element of the basis is orthogonal to $\Pi_i$, the image of every other element is also orthogonal to the image of $\Pi_i$ so we can remove $i$ from the domain basis set and remove $j$ from the codomain basis set and proceed inductively.
  This constructs a permutation $\sigma$ acting on $I$ by mapping each $i$ to a corresponding element $j$.

  % \sigma \in \Sigma
  Now we must show that $\sigma$ is in $\Sigma$.
  The dimensions of $\mathcal{A}$ and $\comm(\mathcal{A})$ as algebras within an element $i$ of the decomposition are $D_i^2 = \dim (\Pi_i \mathcal{A} \Pi_i)$ and $E_i^2 = \dim(\Pi_i \comm(\mathcal{A}) \Pi_i)$.
  We can relate the dimensions within $i$ and $j$ as follows,
  \begin{align}
      \Pi_j \mathcal{A} \Pi_j
      &= \Pi_j U \mathcal{A} U^\dagger \Pi_j \\
      &= (U \Pi_i U^\dagger) U \mathcal{A} U^\dagger (U \Pi_i U^\dagger) \\
      &= U (\Pi_i \mathcal{A} \Pi_i) U^\dagger
  \end{align}
  which has the same dimension as $\Pi_i \mathcal{A} \Pi_i$ since $U$ is unitary.
  % Since $\Ad_{U}$ preserves the dimension of $\mathcal{A}$ and $\mathcal{A'}$ as algebras, we have that $D_i^2 = D_j^2$ and $E_i^2 = E_j^2$.
  Hence, $D_i^2 = D_j^2$ and $E_i^2 = E_j^2$.
  Therefore, the permutation is in the appropriate subgroup $\Sigma$ defined earlier.

  % Proof of Uniqueness
  Uniqueness of this representative element follows because the only element of $\Sigma$ which leaves $Z(\mathcal{A})$ fixed is the identity element.
\end{proof}

\begin{lemma}
  The set of representatives $V$ decomposes as $V \cong \oplus_{i\in I} V_i$.
  \label{lemma:V_direct_sum}
\end{lemma}
\begin{proof}
  For any $\xi$ be some vector state of $\mathcal{A}$ and $v$ an element of the $V$.
  Let $\pi_i : \mathcal{H} \rightarrow \mathcal{H}_{D_i} \otimes \mathcal{H}_{E_i} $ for $i \in I$ be the corresponding natural projection from the direct sum decomposition of $\mathcal{H}$.
  Let $v_i = \pi_i v \pi_i^{-1}$ denote a block of $v$.
  Similarly, we may decompose $\xi$ into a sum of states from these subspaces,
  \begin{align}
    \xi = \bigoplus_{i \in I} \xi_i = \sum_{i \in I} \pi_i^{-1} \xi_i
    \text{.}
  \end{align}

  First, consider the action of the claimed decomposition,
  \begin{align}
      \left(\bigoplus_{i \in I} v_i \right) \xi = \bigoplus_{i \in I} v_i \xi_i
      \text{.}
  \end{align}
  We now compare this to the action of $v$ itself,
  \begin{align}
    v \xi
    &= \bigoplus_{j \in I} \pi_j v \xi
    = \sum_{i \in I} \bigoplus_{j \in I} \pi_j v \pi_i^{-1} \xi_i \\
    &= \bigoplus_{j \in I} \pi_j v \pi_j^{-1} \xi_j \\
    &= \bigoplus_{i \in I} v_i \xi_i
  \end{align}
  since $\pi_j v \pi_i^{-1} {=} 0$ for $i {\neq} j$ from the previous lemma.
  Hence, the action of $v$ and $\oplus_i v_i$ is equal on any vector state and therefore they are equal.
\end{proof}

\begin{lemma}
  \begin{align}
    V_i = \mathrm{U}_{D_i} \otimes \mathrm{U}_{E_i}
  \end{align}
\end{lemma}
\begin{proof}
    Let $\mathcal{A}_i$ and $\mathcal{A}'_i$ be the associated factors from $\mathcal{A}$ and its commutant respectively.
    As finite-dimensional factors, these are central simple algebras.
    Consequently, by the Skolem-Noether theorem their automorphisms are inner~\cite{FarbDennis2016NoncommutativeAlgebra}.
    Let $a \in \mathcal{A}$, $b \in \mathcal{A}'$ and $v_i \in V_i$.
    This means that there exist unique unitaries $v_{D_i}$ and $v_{E_i}$ such that,
    \begin{align}
        \Ad_{v_i} (a) &= \Ad_{u_{D_i} \otimes \mathds{1}_{E_i}} (a) \\
        \Ad_{v_i} (b) &= \Ad_{\mathds{1}_{D_i} \otimes u_{E_i}} (b)\text{.}
    \end{align}
    This may be extended to all products,
    \begin{align}
        \Ad_{v_i} (ab)
        &= \Ad_{v_i} (a) \Ad_{v_i} (b) \\
        &= (u_{D_i} \otimes u_{E_i}) a b (u_{D_i} \otimes u_{E_i})^\dagger
    \end{align}
    and to all linear combinations.
    Hence $V_i$ coincides with $U_{D_i} \otimes U_{E_i}$ on the product algebra $\mathcal{A}_i \mathcal{A}'_i$.
    However, since this is the full matrix algebra for the space in which it acts, $v_i$ has been fully determined.
\end{proof}

\begin{theorem}[Structure theorem for the unitary automorphisms]
  \begin{align}
    \mathrm{U}(\mathcal{A}) \cong \left( \bigoplus_{i \in I} \left(\mathrm{U}_{D_i} \otimes \mathrm{U}_{E_i} \right)\right) \rtimes \Sigma
  \end{align}
  where $\Sigma$ is the automorphism group for the factors of $\mathcal{A}$.
  and the semidirect product is defined by the right-action,
  \begin{align}
    \left(\bigoplus_{i \in I} v_i \otimes v_i\right) \triangleleft \sigma = \sum_{i \in I} v_{\sigma^{-1}(i)} \times v_{\sigma^{-1}(i)}
  \end{align}
  which simply rearranges equivalent representations of $\mathcal{A}$.
\end{theorem}

% CJT: This proof is more elementary
% \begin{proof}
% The correspondence as sets should be clear from the preceding lemmas; what remains is to derive the group product on the decomposition.

% For any $U_a, U_b \in \mathrm{U}(\mathcal{A})$ there exist decompositions $U_a = a \sigma_a$ and $U_b = b \sigma_b$ with some $a,b \in V$ and $\sigma_a,\sigma_b \in \Sigma$.
% Now we examine the group product,
% \begin{align}
% U_a U_b
% &= (a \sigma_a) (b \sigma_b) \\
% &= (a \sigma_a b \sigma_a^{-1}) (\sigma_a \sigma_b)\text{.}
% \end{align}

% The first factor $a \sigma_a b \sigma_a^{-1}$ is in $V$ because if we start with a basis element $\Pi_i$ of $Z(\mathcal{A})$ this is first mapped to $\Pi_{\sigma_a^{-1}(i)}$ by $\sigma_a^{-1}$ then left fixed by $b$, only to be mapped back to $\Pi_i$ by $\sigma_a$ and finally left fixed again by the action of $a$.

% Clearly $\sigma_a \sigma_b$ is in $\Sigma$.
% Therefore, we have defined a group structure on the cartesian product $V \times \Sigma$ which takes the form of a semi-direct product. In particular, it is the one in the theorem statement.
% \end{proof}

% This proof assumes a little more algebra but is more illuminating.
\begin{proof}
Let $v \in V$ and $U \in U(\mathcal{A})$.
By \Cref{lemma:U_unique_decomposition}, there exists unique $a \in V$ and $\sigma \in \Sigma$ such that $U = a \sigma$.
We simply verify normality,
\begin{align}
  (a\sigma) v (a\sigma)^{-1} &= a (\sigma v \sigma^{1}) a^{-1} \\
  &= a v' a^{-1} \in V\text{,}
\end{align}
where $v' \in V$ after the second equality by considering the action of $\sigma$ on the structure of $v$ from \Cref{lemma:V_direct_sum}.
Hence, $V$ is normal in $U(\mathcal{A})$.
In combination with the unique decomposition, this establishes that $U(\mathcal{A}) \cong V \rtimes \Sigma$.

Combining this result with the preceding lemmas will derive the claimed group structure; and examination of the conjugation action shows that it is the one stated in the theorem.
\end{proof}

\section{Conditional unitaries \label{app:conditionals}}

\begin{definition}[Conditional unitaries]
    Let $V_{LC}$ be a bi-partite unitary map on $\mathcal{H}_{LC}$ and $Q_{L} \otimes \Id_C$ a local Hermitian operator with non-degenerate spectrum. $V$ is a $Q$-conditional unitary if $[V,Q] = 0$.
\end{definition}
\begin{lemma}
     $V_{LC}$ is a $Q_L$-conditional unitary if and only if it can be written as:
    \begin{equation}
        V_{LC} = \sum_{i} \Pi^{(i)}_L \otimes \xi^{(i)}_C,
    \end{equation}
    \noindent where $\Pi^{(i)}_L$ denotes the spectral projectors of $Q$, and $\xi^{(i)}$ are arbitrary local unitaries and $i$ indexes the unique eigenvalues of $Q$.
\end{lemma}
\begin{proof}
For the forward implication, substitute the form of operators into the commutator equation (with dropped subpscripts for convenience):
\begin{align}
        [V, Q] &= \left [  \sum_i \Pi^{(i)} \otimes \xi^{(i)}, \sum_j q_j\Pi^{(j) }  \otimes \Id\right] \\
            &= \sum_{i,j} q_j [\Pi^{(i)}, \Pi^{(j)}] \otimes \xi^{(i)} \\
            &= 0,
\end{align}
\noindent where we utilised the orthogonality relation $\Pi^{(i)}\Pi^{(j)} = \delta_{ij}\Pi^{(j)}$.

For the reverse implication, write the commutator $VQ - QV$ in the local $Q$-eigenbasis: 
\begin{align}
    [V,Q]  &=\sum_{a,b} (\ket{a}\bra{b}Q - Q\ket{a} \bra{b}) \otimes \chi^{(a,b)}, \\ 
    &= \sum_{a,b} (q_b \ket{a}\bra{b}-  q_{a}\ket{a} \bra{b}) \otimes \chi^{(a,b)}, \\
    &= \sum_{a,b} (q_b-q_a)\ket{a}\bra{b}\otimes \chi^{(a,b)}.
\end{align}
\noindent By assumption, each eigenvalue of $Q$ is unique hence each term in the sum is linearly independent from the others. Then $[V,Q]= 0$ demands $\chi^{(a,b)} = \delta_{ab} \xi^{(a,a)}$ and unitarity fixes each $\xi^{(a,a)}$ to be a local unitary. This gives the required form of $V$.
\end{proof}
Conditional unitaries are ones which have a local conserved observable $Q$. They act on the target space $C$ conditional on classical outcome of measuring in the eigenbasis of $Q$ on the control space $L$. 
\begin{corollary}
    Let $V_{LC}$ be a $Q_L$-conditional unitary. Then, for any local operator $\Id_L \otimes a_C$ and $t \geq 0$, we have $[\Ad_V^t(\Id \otimes a_C), Q_L] = 0$ for all $t \geq 0$.
\end{corollary}
\begin{proof}
    From the previous lemma, we have:
    \begin{equation}
        V^t = \sum_i \Pi^{(i)} \otimes (\xi^{(i)})^t,
    \end{equation}
    \noindent from the idempotence of spectral projectors. Then for the time evolution of $a$, one obtains: 
    \begin{equation}
        \Ad_V^t(\Id \otimes a) = \sum_i \Pi^{(i)} \otimes \Ad^t_{\xi^{(i)}}(a),
    \end{equation}
    \noindent which clearly commutes with $Q$ since it is diagonal in the same local basis.
\end{proof}
The above corollary highlights the importance of conditional unitaries for causal decoupling, as any local operator from the target support develops support on the control space to commute with the conserved charge of the unitary. This reduces the full variety of operators to the Abelian algebra spanned $ \langle \Pi^{(i)}\rangle_i $ on the control space. This property makes this class of unitaries a convenient workhorse to construct simple examples of wall unitaries in \Cref{sec:rep_theory}.

\whencolumns{}{\newpage}

\bibliography{export}

\end{document}